\documentclass[letterpaper,twocolumn,10pt]{article}
\usepackage{usenix-2020-09}

\usepackage{adjustbox}
\usepackage{algorithm}
\usepackage{algorithmicx}
\usepackage{algpseudocode}
\usepackage{amsthm}
\usepackage{amsmath} 
\usepackage{amssymb}
\usepackage{amsfonts}
\usepackage{appendix}
\usepackage{array}

\usepackage{booktabs}

\usepackage{caption} 
\usepackage{cite}
\usepackage{colortbl}

\usepackage{dsfont}

\usepackage{filecontents}

\usepackage{graphicx} 

\usepackage{listings}
\usepackage{longtable}

\usepackage{multicol}
\usepackage{multirow}

\usepackage{pifont}

\usepackage{setspace}
\usepackage{stfloats}
\usepackage{subcaption}  

\usepackage{textcomp}
\usepackage{tikz} 
\usepackage{titlesec}

\usepackage{ulem}

\usepackage{xcolor}

\usepackage{epstopdf}

\usepackage{tikz-cd}

\newtheorem{theorem}{Theorem}
\newtheorem{lemma}[theorem]{Lemma}
\newtheorem{corollary}[theorem]{Corollary}
\newtheorem{definition}{Definition}
\newtheorem{proposition}[theorem]{Proposition}

\newcommand{\ignore}[1]{}

\DeclareMathOperator*{\argmax}{arg\,max}

\begin{document}

\date{}
\title{\Large \bf Understanding Impacts of Task Similarity on Backdoor Attack and Detection}

\author{
{\rm Di Tang}\\
Indiana University Bloomington
\and
{\rm Rui Zhu}\\
Indiana University Bloomington
\and
{\rm XiaoFeng wang}\\
Indiana University Bloomington
\and
{\rm Haixu Tang}\\
Indiana University Bloomington
\and
{\rm Yi Chen}\\
Indiana University Bloomington
} 

\maketitle

\thispagestyle{plain}
\pagestyle{plain}

\begin{abstract}

With extensive studies on backdoor attack and detection, still fundamental questions are left unanswered regarding the limits in the adversary's capability to attack and the defender's capability to detect. We believe that answers to these questions can be found through an in-depth understanding of the relations between the \textit{primary task} that a benign model is supposed to accomplish and the \textit{backdoor task} that a backdoored model actually performs. For this purpose, we leverage similarity metrics in multi-task learning to formally define the backdoor distance (similarity) between the primary task and the backdoor task, and analyze existing stealthy backdoor attacks, 
revealing that most of them fail to effectively reduce the backdoor distance and even for those that do, still much room is left to further improve their stealthiness. 
So we further design a new method, called TSA attack, to automatically generate a backdoor model under a given distance constraint, and demonstrate that our new attack indeed outperforms existing attacks, making a step closer to understanding the attacker's limits. Most importantly, we provide both theoretic results and experimental evidence on various datasets for the positive correlation between the backdoor distance and backdoor detectability, demonstrating that indeed our task similarity analysis help us better understand backdoor risks and has the potential to identify more effective mitigations. 

\end{abstract}

\vspace{-10pt}
\section{Introduction}
\vspace{-5pt}


A backdoor is a function hidden inside a machine learning (ML) model, through which a special pattern on the model's input, called a \textit{trigger}, can induce misclassification of the input. The backdoor attack is considered to be a serious threat to trustworthy AI, allowing the adversary to control the operations of an ML model, a deep neural network (DNN) in particular, for the purposes such as evading malware detection~\cite{DBLP:conf/uss/SeveriMCO21}, gaming a facial-recognition system to gain unauthorized access~\cite{reflection}, etc. 

\vspace{3pt}\noindent\textbf{Task similarity analysis on backdoor}. With continued effort on backdoor attack and detection, this emerging threat has never been fully understood. 
Even though new attacks and detections continue to show up, they are mostly responding to some specific techniques, and therefore offer little insights into the best the adversary could do and the most effective strategies the detector could possibly deploy.    

Such understanding is related to the similarity between the \textit{primary task} that a benign model is supposed to accomplish and the \textit{backdoor task} that a backdoored model actually performs, which is fundamental to distinguishing between a backdoored model and its benign counterpart. Therefore, a \textit{Task Similarity Analysis} (\textit{TSA}) between these two tasks can help us calibrate the extent to which a backdoor is detectable (differentiable from a benign model) by not only known but also new detection techniques, inform us which characters of a backdoor trigger contribute to the improvement of the similarity, thereby making the attack stealthy, and further guides us to develop even stealthier backdoors so as to better understand what the adversary could possibly do and what the limitation of detection could actually be. 

\vspace{3pt}\noindent\textbf{Methodology and discoveries}. 
This paper reports the first TSA on backdoor attacks and detections. We formally model the backdoor attack and define \textit{backdoor similarity} based upon the task similarity metrics utilized in multi-task learning to measure the similarity between the backdoor task and its related primary task. 
On top of the metric, we further define the concept of $\alpha$-backdoor to compare the backdoor similarity across different backdoors, and present a technique to estimate the $\alpha$ for an attack in practice. With the concept of $\alpha$-backdoor, we analyze representative attacks proposed so far to understand the stealthiness they intend to achieve, based upon their effectiveness in increasing the backdoor similarity. 
We find that current attacks only marginally increased the overall similarity between the backdoor task and the primary tasks, due to that they failed to simultaneously increase the similarity of inputs and that of outputs between these two tasks.
Based on this finding, we develop a new attack/analysis technique, called \textit{TSA attack}, to automatically generate a backdoored model under a given similarity constraint. The new technique is found to be much stealthier than existing attacks, not only in terms of backdoor similarity, but also in terms of its effectiveness in evading existing detections, as observed in our experiments.  
Further, we demonstrate that the backdoor with high backdoor similarity is indeed hard to detect through theoretic analysis as well as extensive experimental studies on four datasets under six representative detections using our TSA attack together with five representative attacks proposed in prior researches.

\vspace{3pt}\noindent\textbf{Contributions}. Our contributions are as follows: 

\vspace{2pt}\noindent$\bullet$\textit{ New direction on backdoor analysis}. Our research has brought a new aspect to the backdoor research, through the lens of backdoor similarity. Our study reveals the great impacts backdoor similarity has on both backdoor attack and detection, which can potentially help determine the limits of the adversary's capability in a backdoor attack and therefore enables the development of the best possible response.

\vspace{2pt}\noindent$\bullet$\textit{ New stealthy backdoor attack}. 
Based upon our understanding on backdoor similarity, we developed a novel technique, TSA attack, to generate a stealthy backdoor under a given backdoor similarity constraint, helping us better understand the adversary's potential and more effectively calibrate\ignore{ 
achieve full control of the backdoor and making it possible to calibrate} the capability of backdoor detections,  




\vspace{-10pt}
\section{Background}
\label{sec: background}

\vspace{-5pt}
\subsection{Neural Network}
\label{subsec:nn_modeling}
We model a neural network model $f$ as a mapping function from the input space $\mathcal{X}$ to the output space $\mathcal{Y}$, i.e., $f:\mathcal{X} \mapsto \mathcal{Y}$. Further, the model $f$ can be decomposed into two sub-functions: $f(x) = c(g(x))$. Specifically, for a classification task with $L$ classes where the output space $\mathcal{Y} = \{0, 1, ... ,L-1\}$, we define $g:\mathcal{X} \mapsto [0,1]^L$, $c:[0,1]^L \mapsto \mathcal{Y}$ and $c(g(x)) = \argmax_j g(x)_j$ where $g(x)_j$ is the $j$-th element of $g(x)$. According to the common understanding, after well training, $g(x)$ approximates the conditional probability of presenting $y$ given $x$, i.e., $g(x)_y \approx Pr(y|x)$, for $y \in \mathcal{Y}$ and $x \in \mathcal{X}$.

\vspace{-5pt}
\subsection{Backdoor Attack \& Detection}
\label{subsec:attack_detection}
\noindent\textbf{Backdoor attack}.
In our research, we focus on targeted backdoors that cause the backdoor infected model $f_b$ to map trigger-carrying inputs $A(x)$ to the target label $t$ different from the ground truth label of $x$~\cite{BlindBackdoor, Backdoor_in_poison, NeuralCleanse, TrojanMeta}:
\vspace{-10pt}
\begin{equation}
\begin{array}{r@{\quad}l}
& f_b(A(x)) = t \neq f_P(x) \\
\end{array}
\vspace{-10pt}
\label{eq:backdoor_current_goal}
\end{equation}
where $f_P$ is the benign model that outputs the ground truth label for $x$ and $A$ is the trigger function that transfers a benign input to its trigger-carrying counterpart. 
There are many attack methods have been proposed to inject backdoors, e.g., ~\cite{reflection, SCAn, TargetedBackdoor,TrojanAttck, TBT, sin2, weight_perturbations, DBLP:conf/globalsip/ClementsL18}. 

\vspace{3pt}\noindent\textbf{Backdoor detection}. The backdoor detection has been extensively studied recently~\cite{li2022backdoor, DBLP:journals/corr/abs-2007-10760, DBLP:journals/ijon/KavianiS21, DBLP:journals/corr/abs-2202-07183, DBLP:journals/corr/abs-2202-06862}. 
These proposed approaches can be categorized based upon their focuses on different model information: model outputs, model weights and model inputs. This categorization has been used in our research to analyze different detection approaches (Section~\ref{sec:detection}).

More specifically, detection on model outputs captures backdoored models through detecting the difference between the outputs of backdoored models and benign models on some inputs.
Such detection methods include NC~\cite{NeuralCleanse}, K-ARM~\cite{K-ARM}, MNTD~\cite{MNTD}, Spectre~\cite{Spectre}, TABOR~\cite{TABOR}, MESA~\cite{MESA}, STRIP~\cite{STRIP}, SentiNet~\cite{SentiNet}, ABL~\cite{ABL},  ULP~\cite{ULPs}, etc. 
Detection of model weights finds a backdoored model through distinguishing its model weights from those of benign models. Such detection approaches include ABS~\cite{ABS},  ANP~\cite{ANP}, NeuronInspect~\cite{NeuronInspect}, etc. 
Detection of model inputs identifies a backdoored model through detecting difference between inputs that let a backdoored model and a benign model output similarly.
Prominent detections in this category include SCAn~\cite{SCAn}, AC\cite{AC}, SS~\cite{SS}, etc. 


\vspace{-10pt}
\subsection{Threat Model}
\label{subsec:threat_model}
\vspace{-5pt}

We focus on backdoors for image classification tasks, while assuming a white-box attack scenario where the adversary can access the training process. The attacker inject the backdoor to accomplish the goal formally defined in Section~\ref{subsec:backdoor_similarity} and evade from backdoor detections.

The backdoor defender aim to distinguish backdoored models from benign models. She can white-box access those backdoored models and owns a small set of benign inputs. Besides, the defender may obtain a set of mix inputs containing a large number of benign inputs together with a few trigger-carrying inputs, however which inputs carried the trigger in this set is unknown to her. 
\vspace{-10pt}
\section{TSA on Backdoor Attack}
\label{sec:tsa_on_attack}
\vspace{-5pt}

Not only does a backdoor attack aim at inducing misclassification of trigger-carrying inputs to a victim model, but it is also meant to achieve high stealthiness\ignore{ and evasiveness} against backdoor detections. For this purpose, some attacks~\cite{TrojanAttck,LIRA} reduce the $L_p$-norm of the trigger, i.e., $\|A(x)-x\|_p$, to make trigger-carrying inputs be similar to benign inputs, while some others construct the trigger using benign features~\cite{composite-backdoor,FaceHack}. All these tricks are designed to evade specific detection methods. 
Still less clear is the stealthiness guarantee that those tricks can provide against other detection methods. Understanding such stealthiness guarantee requires to model the \textit{detectability} of backdoored models, which depends on measuring \textit{fundamental differences} between backdoored and benign models that was not studied before.

To fill in this gap, we analyze the difference between the task a backdoored model intends to accomplish (called \textit{backdoor task}) and that of its benign counterpart (called \textit{primary task}), which indicates the detectability of the backdoored model, as demonstrated by our experimental study (see Section~\ref{sec:detection}). 
Between these two tasks, we define the concept of \textit{backdoor similarity} -- the similarity between the primary and the backdoor task, by leveraging the task similarity metrics used in multi-task learning studies, and further demonstrate how to compute the similarity in practice. Applying the metric to existing backdoor attacks, we analyze their impacts on the backdoor similarity, which consequently affects their stealthiness against detection techniques (see Section~\ref{sec:detection}). We further present a new algorithm that automatically generates a backdoored model under a desirable backdoor similarity, which leads to a stealthier backdoor attack.   



\vspace{-10pt}
\subsection{Task Similarity}



Backdoor detection, essentially, is a problem about how to differentiate between a legitimate task (\textit{primary task}) a model is supposed to perform and the compromised task (\textit{backdoor task}), which involves backdoor activities, the backdoored model actually runs. To this end, a detection mechanism needs to figure out the difference between these two tasks. According to modern learning theory~\cite{murphy2012machine}, a task can be fully characterized by the distribution on the graph of the function~\cite{bridges1998foundations} -- a joint distribution on the input space $\mathcal{X}$ and the output space $\mathcal{Y}$.
Formally, a task $\mathcal{T}$ is characterized by the joint distribution $\mathcal{D}_{\mathcal{T}}$: $\mathcal{T} := \mathcal{D}_{\mathcal{T}}(\mathcal{X},\mathcal{Y}) = \{\Pr_{\mathcal{D}_{\mathcal{T}}}(x,y): (x,y) \in \mathcal{X} \times \mathcal{Y} \}$. Note that, for a well-trained model $f= c \circ g $ (defined in Section~\ref{subsec:attack_detection}) for task $\mathcal{T}$, we have $g(x)_y \approx \Pr_{\mathcal{D}_\mathcal{T}}(y|x)$ for all $(x,y) \in \mathcal{X} \times \mathcal{Y}$.



With this task modeling,  the mission of backdoor detection becomes how to distinguish the distribution of a backdoor task from that of its primary task. The Fisher's discriminant theorem~\cite{mclachlan2005discriminant} tells us that two distributions become easier to distinguish when they are less similar in terms of some distance metrics, indicating that the \textit{distinguishability} (or separability) of two tasks is positively correlated with their distance. This motivates us to measure the distance between the distributions of two tasks. 
For this purpose, we define the $d_{\mathcal{H}-W1}$ distance, which covers both Wasserstein-1 distance and H-divergence, two most common distance metrics for distributions. 

\begin{definition}[$d_{\mathcal{H}-W1}$ distance]~\label{def:task_distance}
For two distributions $\mathcal{D}$ and $\mathcal{D}'$ defined on $\mathcal{X} \times \mathcal{Y}$, $d_{\mathcal{H}-W1}(\mathcal{D}, \mathcal{D}')$ measures the distance between them two as:
\vspace{-5pt}
\begin{equation}
\begin{array}{r@{\quad}l}
d_{\mathcal{H}-W1}(\mathcal{D},\mathcal{D'}) = \underset{h \in \mathcal{H}}{\sup} [\mathbb{E}_{\Pr_{\mathcal{D}}(x,y)}h(x, y) - \mathbb{E}_{\Pr_{\mathcal{D}'}(x,y)}h(x, y)],
\end{array}
\vspace{-5pt}
\label{eq:h-w1}
\end{equation}
where $\mathcal{H}=\{h: \mathcal{X} \times \mathcal{Y} \mapsto [0,1]\}$.
\end{definition}
\vspace{-8pt}
\begin{proposition}~\label{prop:h-w1}
\begin{equation}
\begin{array}{c}
0 \leq d_{\mathcal{H}-W1}(\mathcal{D},\mathcal{D'}) \leq 1, \\
d_{W1}(\mathcal{D},\mathcal{D'}) \leq d_{\mathcal{H}-W1}(\mathcal{D},\mathcal{D'}) = \frac{1}{2} d_{\mathcal{H}}(\mathcal{D},\mathcal{D'}),
\end{array}
\end{equation}
where $d_{W1}(\mathcal{D},\mathcal{D'})$ is the \textit{Wasserstein-1 distance}~\cite{w-distance} between $\mathcal{D}$ and $\mathcal{D}'$, and $d_{\mathcal{H}}(\mathcal{D},\mathcal{D'})$ is their \textit{$\mathcal{H}$-divergence}~\cite{h-divergence}.
\end{proposition}
\begin{proof}
See Appendix~\ref{proof:h-w1}
\vspace{-5pt}
\end{proof}

\noindent Proposition~\ref{prop:h-w1} shows that $d_{\mathcal{H}-W1}$ is representative: it is the upper-bound of the \textit{Wasserstein-1 distance} and the half of the \textit{$\mathcal{H}$-divergence}.
More importantly, $d_{\mathcal{H}-W1}$ can be easily computed: the optimal function $h$ in Eq.~\ref{eq:h-w1} that maximally separate two distributions can be approximated with a neural network to distinguish them. 

Using the $d_{\mathcal{H}-W1}$ distance, we can now quantify the similarity between tasks. In particular, $d_{\mathcal{H}-W1}(\mathcal{D}_{\mathcal{T}1}, \mathcal{D}_{\mathcal{T}2}) = 0$ indicates that tasks $\mathcal{T}1$ and $\mathcal{T}2$ are identical, and $d_{\mathcal{H}-W1}(\mathcal{D}_{\mathcal{T}1}, \mathcal{D}_{\mathcal{T}2}) = 1$ indicates that these two tasks 
are totally different. Without further notice, we consider the task similarity between $\mathcal{T}1$ and $\mathcal{T}2$ as $1-d_{\mathcal{H}-W1}(\mathcal{D}_{\mathcal{T}1}, \mathcal{D}_{\mathcal{T}2})$.

\vspace{-10pt}
\subsection{Backdoor Similarity}
\label{subsec:backdoor_similarity}

Following we first define primary task and backdoor task and then utilize $d_{\mathcal{H}-W1}$ to specify \textit{backdoor similarity}, that is, the similarity between the primary task and the backdoor task.


\vspace{3pt}\noindent\textbf{Backdoor attack}. 
As mentioned earlier (Section~\ref{subsec:attack_detection}), the well-accepted definition of the backdoor attack is specified by Eq.~\ref{eq:backdoor_current_goal}~\cite{SCAn, BlindBackdoor, Backdoor_in_poison, NeuralCleanse, TrojanMeta, SentiNet, TrojanZOO}. According to the definition, the attack aims to find a trigger function $A(\cdot)$ that maps benign inputs to their trigger-carrying counterparts and also ensures that these trigger-carrying inputs are misclassfied to the target class $t$ by the backdoor infected model $f_b$. In particular, Eq.~\ref{eq:backdoor_current_goal} requires the target class $t$ to be different from the source class of the benign inputs\ignore{ to trigger the malicious consequence}, i.e., $t \neq f_P(x)$. 
This definition, however, is problematic, since there exists a trivial trigger function satisfying  Eq.~\ref{eq:backdoor_current_goal}, i.e., $A(\cdot)$ simply replaces a benign input $x$ with another benign input $x_t$ in the target class $t$. Under this trigger function, even a completely clean model $f_P(\cdot)$ becomes ``backdoored'', as it outputs the target label on any ``trigger-carrying'' inputs $x_t = A(x)$. 

Clearly, this trivial trigger function does not introduce any meaningful backdoor to the victim model, even though it satisfies Eq.~\ref{eq:backdoor_current_goal}.
To address this issue, we adjust the objective of the backdoor attack (Eq.~\ref{eq:backdoor_current_goal}) as follows: 
\vspace{-5pt}
\begin{equation}
\begin{array}{r@{\quad}l}
f_b(A(x)) = t \text{, where } f_P(x) \neq t \neq f_P(A(x)) .\\
\end{array}
\vspace{-5pt}
\label{eq:backdoor_our_goal}
\end{equation}
Here, the constraint $f_P(x) \neq t \neq f_P(A(x))$ requires that under the benign model $f_P$, not only the input $x$ but also its trigger-carrying version $A(x)$ will not be mapped to the target class $t$, thereby excluding the trivial attack mentioned above. 

Generally speaking, the trigger function $A(\cdot)$ may not work on a model's whole input space.  So we introduce the concept of \textit{backdoor region}:
\begin{definition}[Backdoor region]\label{def:backdoor_region}
The backdoor region $\mathcal{B} \subset \mathcal{X}$ of a backdoor with the trigger function $A(\cdot)$ is the set of inputs on which the backdoored model $f_b$ satisfy Eq.~\ref{eq:backdoor_our_goal}, i.e.,
\vspace{-5pt}
\begin{equation}
\begin{array}{r@{\quad}l}
f_b(A(x)) = 
\begin{cases}
    t, \neq f_P(A(x)), \neq f_P(x),      & \forall x \in \mathcal{B} \\
    f_P(A(x)), & \forall x \in \mathcal{X} \setminus \mathcal{B}.
\end{cases} 
\end{array}
\vspace{-3pt}
\label{eq:backdoor_region}
\end{equation}
Accordingly, we denote $A(B)=\{A(x): x \in B\}$ as the set of trigger-carrying inputs.
\end{definition}
\noindent For example, the backdoor region of a source-agnostic backdoor, which maps the trigger-carrying input $A(x)$ whose label under the benign model is not $t$ into $t$, is $\mathcal{B} = \mathcal{X} \setminus (\mathcal{X}_{f_P(x)=t}\cup \mathcal{X}_{f_P(A(x))=t})$, while the backdoor region for a source-specific backdoor, which maps the trigger-carrying input $A(x)$ with the true label of the source class $s$ ($\neq t$) into $t$, is $\mathcal{B} = \mathcal{X}_{f_P(x)=s} \setminus \mathcal{X}_{f_P(A(x))=t}$. 
Here, we use $\mathcal{X}_{C}$ to denote the subset of all elements in $\mathcal{X}$ that satisfy the condition $C$:  
$\mathcal{X}_{\textit{C}} = \{x | x \in \mathcal{X}, \textit{C is True} \}$, 
e.g., $\mathcal{X}_{f_P(x)=t} = \{x| x \in \mathcal{X}, f_P(x)=t \}$. 

\vspace{3pt}\noindent\textbf{Definition of the primary and backdoor tasks}.
Now we can formally define the primary task and the backdoor task for a backdoored model.
Here we denote the prior probability of input $x$ (also the probability of presenting $x$ on the primary task) by $\Pr(x)$.
 
\vspace{-5pt}
\begin{definition} [Primary task $\&$ distribution]
The primary task of a backdoored model is $\mathcal{T}_P$, the task that its benign counterpart learns to accomplish. $\mathcal{T}_P$ is characterized by the primary distribution $\mathcal{D}_P$, a joint distribution over the input space $\mathcal{X}$ and the output space $\mathcal{Y}$. Specifically, $\Pr_{\mathcal{D}_P}(x,y)$ is the probability of presenting $(x,y)$ in benign scenarios, and thus $\Pr_{\mathcal{D}_P}(y|x) = {\Pr_{\mathcal{D}_P}(x,y)}/{\Pr(x)}$ is the conditional probability that a benign model strives to approximate. 
\end{definition} 

\vspace{-5pt}
\begin{definition} [Backdoor task $\&$ distribution]
\label{def:backdoor_task}
The backdoor task of a backdoored model is denoted by $\mathcal{T}_{A,\mathcal{B},t}$, the task that the adversary intends to accomplish by training a backdoored model. $\mathcal{T}_{A,\mathcal{B},t}$ is characterized by the backdoor distribution $\mathcal{D}_{A,\mathcal{B},t}$, a joint distribution over $\mathcal{X} \times \mathcal{Y}$.
Specifically, the probability of presenting $(x,y)$ in $\mathcal{D}_{A,\mathcal{B},t}$ is $\Pr_{\mathcal{D}_{A,\mathcal{B},t}}(x,y)=P(x,y)/Z_{A,\mathcal{B},t}$, where $Z_{A,\mathcal{B},t} = \int_{(x,y) \in \mathcal{X}\times\mathcal{Y}}P(x,y) = 1-\Pr(A(\mathcal{B}))+\beta \Pr(\mathcal{B}) $ and
\vspace{-5pt}
\begin{equation}
\begin{array}{r@{\quad}l}
P(x,y) = 
\begin{cases}
    \Pr_{\mathcal{D}_{A,\mathcal{B},t}}(y|x) \Pr(A^{-1}(x) ) \beta, & x \in A(\mathcal{B}) \\
    \Pr_{\mathcal{D}_P}(x,y),              & x \in \mathcal{X} \setminus A(\mathcal{B}).
\end{cases}
\end{array}
\vspace{-3pt}
\label{eq:backdoor_distribution}
\end{equation}
\end{definition}

\noindent Here, $A^{-1}(x)=\{z|A(z)=x\}$ represents the inverse of the trigger function, $\Pr_{\mathcal{D}_{A,\mathcal{B},t}}(y|x)$ is the conditional probability that the adversary desires to train a backdoored model to approximate, $\beta$ is a parameter selected by the adversary to amplify the probability that the trigger-carrying inputs $A(x)$ are presented to the backdoor task. Actually, we consider $\frac{\beta}{1+\beta}$ as the \textit{poisoning rate} with the assumption that poisoned training data is randomly drawn from the backdoor distribution. Finally, it is worth noting that $\Pr_{\mathcal{D}_{A,\mathcal{B},t}}(x,y)$ is proportional to $\Pr_{\mathcal{D}_P}(x,y)$ except on those trigger-carrying inputs $A(\mathcal{B})$.


\vspace{3pt}\noindent\textbf{Formalization of backdoor similarity}.
Putting together the definitions of the primary task, the backdoor task, and the $d_{\mathcal{H}-W1}$ distance between the two tasks (Eq.~\ref{eq:h-w1}), we are ready to define \textit{backdoor similarity} as follows:

\begin{definition}[Backdoor distance \& similarity]~\label{def:backdoor_distance} 
We define $d_{\mathcal{H}-W1}(\mathcal{D}_P, \mathcal{D}_{A,\mathcal{B},t})$ as the backdoor distance between the primary task $\mathcal{T}_P$ and the backdoor task $\mathcal{T}_{A,\mathcal{B},t}$ and $1 - d_{\mathcal{H}-W1}(\mathcal{D}_P, \mathcal{D}_{A,\mathcal{B},t})$ as the backdoor similarity 

\end{definition}

\begin{theorem}[Computing backdoor distance]~\label{thr:backdoor_general_distance}
When $Z_{A,\mathcal{B},t} \geq 1$, where $Z_{A,\mathcal{B},t}$ is defined in Eq.~\ref{eq:backdoor_distribution}, the backdoor distance between $\mathcal{D}_P$ and $\mathcal{D}_{A,\mathcal{B},t}$ is
\vspace{-10pt}
\begin{equation}
\notag
\begin{array}{l@{\quad}l}
d_{\mathcal{H}-W1}(\mathcal{D}_P, \mathcal{D}_{A,\mathcal{B},t}) = \underset{(x,y)\in A(\mathcal{B})\times\mathcal{Y}}{\int} \max(\Pr_{gain}(x,y), 0), \\
\end{array}
\vspace{-5pt}
\end{equation}
\noindent where 
\vspace{-10pt}
\begin{equation}
\notag
\begin{array}{l@{\quad}l}
\Pr_{gain}(x,y) =\Pr_{\mathcal{D}_{A,\mathcal{B},t}}(x,y) - \Pr_{\mathcal{D}_P}(x,y).
\end{array}
\vspace{-5pt}
\end{equation}

\end{theorem}

\begin{proof}
See Appendix~\ref{proof:backdoor_general_distance}.
\vspace{-8pt}
\end{proof}

\noindent Theorem~\ref{thr:backdoor_general_distance} shows that the calculation of backdoor distance $d_{\mathcal{H}-W1}(\mathcal{D}_P, \mathcal{D}_{A,\mathcal{B},t})$ can be reduced to the calculation of the probability gain of $\Pr_{\mathcal{D}_{A,\mathcal{B},t}}(x,y)$ over $\Pr_{\mathcal{D}_P}(x,y)$ on those trigger-carrying inputs $A(\mathcal{B})$, when $Z_{A,\mathcal{B},t} \geq 1$. 
Notably, because 
$Z_{A,\mathcal{B},t} = 1-\Pr(A(\mathcal{B})) + \beta \Pr(\mathcal{B})$,  $Z_{A,\mathcal{B},t} \geq 1$ is satisfied if
$\Pr(A(\mathcal{B})) \leq \beta \Pr(\mathcal{B})$. This implies that
if those trigger-carrying inputs show up more often on the 
backdoor distribution than on the primary distribution, we can use the aforementioned method to compute the backdoor distance.


\vspace{3pt}\noindent\textbf{Parametrization of backdoor distance}.
The following Lemma further reveals the impacts of two parameters $\beta$ and $\kappa$ on the backdoor distance:

\begin{lemma}\label{lemma:backdoor_A_beta}
When, $Z_{A,\mathcal{B},t} \geq 1$ and $\Pr(\mathcal{B}) = \kappa \Pr(A(\mathcal{B}))$,
\vspace{-8pt}
\begin{equation}
\notag
\begin{array}{l@{\quad}l}
d_{\mathcal{H}-W1}(\mathcal{D}_P, \mathcal{D}_{A,\mathcal{B},t}) = \Pr(\mathcal{B}) \underset{(x,y)\in A(\mathcal{B})\times\mathcal{Y}}{\int}\max(\widetilde{\Pr_{gain}}(x,y), 0), 
\end{array}
\vspace{-3pt}
\end{equation}
\noindent where $\widetilde{\Pr_{gain}}(x,y)$ equals to
\vspace{-10pt}
\begin{equation}
\label{eq:bain}
\begin{array}{l@{\quad}l}
\frac{\beta}{Z_{A,\mathcal{B},t}} \frac{\Pr_{\mathcal{D}_{A,\mathcal{B},t}}(x)}{\Pr_{\mathcal{D}_{A,\mathcal{B},t}}(A(\mathcal{B}))}\Pr_{\mathcal{D}_{A,\mathcal{B},t}}(y|x) - \frac{1}{\kappa} \frac{\Pr(x)}{\Pr(A(\mathcal{B}))} \Pr_{\mathcal{D}_P}(y|x).
\end{array}
\vspace{-8pt}
\end{equation}
\end{lemma}

\begin{proof}
The derivation is straightforward, thus we omit it.  
\vspace{-8pt}
\end{proof}

As demonstrated by Lemma~\ref{lemma:backdoor_A_beta}, the two parameters $\beta$ and $\kappa$ are important to the backdoor distance, where $\beta$ is related to the poisoning rate (Definition~\ref{def:backdoor_task}) and $\kappa$ describes how close is the probability of presenting trigger-carrying inputs to the probability of showing their benign counterparts on the primary distribution (the bigger $\kappa$ the farther away are these two probabilities).

Let us first consider the range of $\beta$. Intuitive, a large $\beta$ causes the trigger-carrying inputs more likely to show up on the backdoor distribution, and therefore could be easier detected.
A reasonable backdoor attack should keep $\beta$ smaller than $1$, which is equivalent to constraining the poisoning rate ($\frac{\beta}{1+\beta}$) below $50\%$. On the other hand, a very small $\beta$ will make the backdoor task more difficult to learn by a model, which eventually reduces the attack success rate (ASR). A reasonable backdoor attack should use a $\beta$ greater than $\frac{1}{\kappa}$: that is, the chance of seeing trigger-carrying inputs on the backdoor distribution no lower than that on the primary distribution.
Therefore, we assume $\frac{1}{\kappa} \leq \beta \leq 1$.   
Next, we consider the range of $\kappa$. A reasonable lower-bound of $\kappa$ is $1$; if $\kappa < 1$, trigger-carrying inputs show up even more often than their benign counterparts on the primary distribution, which eventually lets the backdoored model outputs differently from benign models on such large portion of inputs and make the backdoor be easy detected.
So, we assume $\kappa \geq 1$.

With above assumptions on the range of $\beta$ and $\kappa$, we get the following theorem to describe the range of backdoor distance.

\begin{theorem}[Backdoor distance range]\label{thr:backdoor_distance_range}
Supposing $\Pr(\mathcal{B})=\kappa \Pr(A(\mathcal{B}))$, when $\kappa \geq 1$ and $\frac{1}{\kappa} \leq \beta \leq 1$, we have $Z_{A,\mathcal{B},t} \geq 1$,
\vspace{-6pt}
\begin{equation}
\notag
\begin{array}{l@{\quad}l}
 (\frac{\beta}{Z_{A,\mathcal{B},t}} - \frac{1}{\kappa}(1- S)) \Pr(\mathcal{B}) \leq d_{\mathcal{H}-W1}(\mathcal{D}_P, \mathcal{D}_{A,\mathcal{B},t}) \leq \frac{\beta}{Z_{A,\mathcal{B},t}} \Pr(\mathcal{B}) ,
\end{array}
\vspace{-5pt}
\end{equation}
\noindent where $S = \underset{(x,y) \in A(\mathcal{B}) \times \mathcal{Y} }{\int} \max\{\Delta_{prob}, 0\}$ and 
\vspace{-5pt}
\begin{equation}
\label{eq:S}
\begin{array}{l@{\quad}l}
\Delta_{prob} = \frac{\Pr_{\mathcal{D}_{A,\mathcal{B},t}}(x)}{\Pr_{\mathcal{D}_{A,\mathcal{B},t}}(A(\mathcal{B}))}\Pr_{\mathcal{D}_{A,\mathcal{B},t}}(y|x) -  \frac{\Pr(x)}{\Pr(A(\mathcal{B}))} \Pr_{\mathcal{D}_P}(y|x).
\end{array}
\vspace{-5pt}
\end{equation}
\end{theorem}

\begin{proof}
See Appendix~\ref{proof:backdoor_distance_range}.
\vspace{-5pt}
\end{proof}

\begin{corollary}[Effects of $\beta$]\label{coro:beta_effect}
Supposing $\Pr(\mathcal{B})=\kappa \Pr(A(\mathcal{B}))$, $\kappa \geq 1$ and $\kappa$ is fixed, when $\beta$ varies in range $[\frac{1}{\kappa}, 1]$, we have
\vspace{-8pt}
\begin{equation}
\notag
\begin{array}{l@{\quad}l}
 \frac{S \Pr(\mathcal{B})}{\kappa}  \leq d_{\mathcal{H}-W1}(\mathcal{D}_P, \mathcal{D}_{A,\mathcal{B},t}) \leq  \frac{\kappa \Pr(\mathcal{B})}{\kappa + \kappa \Pr(\mathcal{B}) - \Pr(\mathcal{B})},
\end{array}
\vspace{-8pt}
\end{equation}
\noindent where $S$ is defined in Theorem~\ref{thr:backdoor_distance_range}. Specially, the lower-bound $\frac{S \Pr(\mathcal{B})}{\kappa}$ is achieved when $\beta=\frac{1}{\kappa}$, and the upper-bound $\frac{\kappa \Pr(\mathcal{B})}{\kappa + \kappa \Pr(\mathcal{B}) - \Pr(\mathcal{B})}$ is achieved when $\beta=1$.
\end{corollary}

\begin{proof}
See Appendix~\ref{proof:beta_effect}.
\vspace{-5pt}
\end{proof}

\begin{corollary}[Effects of $\kappa$]\label{coro:kappa_effect}
Supposing $\Pr(\mathcal{B})=\kappa \Pr(A(\mathcal{B}))$, $\beta \leq 1$ and $\beta$ is fixed, when $\kappa$ varies in range $[\frac{1}{\beta}, \infty)$, we have
\vspace{-10pt}
\begin{equation}
\notag
\begin{array}{l@{\quad}l}
 S \beta \Pr(\mathcal{B}) \leq d_{\mathcal{H}-W1}(\mathcal{D}_P, \mathcal{D}_{A,\mathcal{B},t}) \leq  \beta \Pr(\mathcal{B}),
\end{array}
\vspace{-5pt}
\end{equation}
\noindent where $S$ is defined in Theorem~\ref{thr:backdoor_distance_range}. Specially, the lower-bound $S \beta \Pr(\mathcal{B})$ and the upper-bound $\beta \Pr(\mathcal{B})$ are achieved, respectively, when $\kappa=\frac{1}{\beta}$.
\end{corollary}

\begin{proof}
See Appendix~\ref{proof:kappa_effect}.
\vspace{-10pt}
\end{proof}

\vspace{-10pt}
\subsection{$\alpha$-Backdoor}
\label{subsec:alpha_backdoor}
\vspace{-5pt}

\noindent\textbf{Definition of $\alpha$-backdoor}.
Through Lemma~\ref{lemma:backdoor_A_beta} and Theorem~\ref{thr:backdoor_distance_range}, we show that the backdoor distance and its boundaries are proportional to $\Pr(\mathcal{B})$, the probability of showing benign inputs in the backdoor region $\mathcal{B}$ on the prior distribution of inputs. However, different backdoor attacks may have different backdoor regions, which is a factor we intend to remove so as to compare the backdoor similarities across different attacks. For this purpose, here we define $\alpha$-backdoor, based upon the same backdoor region $\mathcal{B}$ for different attacks, as follows:   


\begin{definition}[$\alpha$-backdoor]\label{def:alpha_backdoor}
We define an $\alpha$-backdoor as a backdoor whose backdoor distribution is $\mathcal{D}_{A,\mathcal{B},t}$, primary distribution is $\mathcal{D}_P$ and the associated backdoor distance equals to the product of $\alpha$ and $\Pr(\mathcal{B})$, i.e., 
\vspace{-8pt}
\begin{equation}
\notag
\begin{array}{r@{\quad}l}
\alpha \cdot \Pr(\mathcal{B}) = d_{\mathcal{H}-W1}(\mathcal{D}_P, \mathcal{D}_{A,\mathcal{B},t}).
\end{array}
\vspace{-5pt}
\end{equation}
\end{definition}

\vspace{3pt}\noindent\textbf{Approximation of $\alpha$}.
Lemma~\ref{lemma:backdoor_A_beta} actually provides an approach to approximate $\alpha$ in practice. 
Specifically, using the symbol $\widetilde{\Pr_{gain}}$ that has been defined in Eq.~\ref{eq:bain}, we get a simple formulation of $\alpha$: $\alpha = \underset{(x,y) \in A(\mathcal{B}) \times \mathcal{Y}}{\int} max(\widetilde{\Pr_{gain}}(x,y),0)$.
Note that 
$\frac{\Pr(x)}{\Pr(A(\mathcal{B}))} = \Pr(x | x\in A(\mathcal{B}))$ 
and $\frac{\Pr_{\mathcal{D}_{A,\mathcal{B},t}}(x)}{\Pr_{\mathcal{D}_{A,\mathcal{B},t}}(A(\mathcal{B}))} = \Pr_{\mathcal{D}_{A,\mathcal{B},t}}(x | x \in A(\mathcal{B}))$. 
This enables us to approximate $\alpha$ through sampling only trigger-carrying inputs $x \in A(\mathcal{B})$.
Also, $\Pr_{\mathcal{D}_{A,\mathcal{B},t}}(y|x)$ and $\Pr_{\mathcal{D}_P}(y|x)$ can be approximated by a well-trained backdoored model $f_b=c_b \circ g_b$ and a well-trained benign model $f_P = c_P \circ g_P$, respectively, i.e., $g_b(x)_y \approx \Pr_{\mathcal{D}_{A,\mathcal{B},t}}(y|x)$ and $g_P(x)_y \approx \Pr_{\mathcal{D}_P}(y|x)$.
Supposing that we have sampled $m$ trigger-carrying inputs $\{A(x_1), A(x_2), ... , A(x_m)\}$, 
$\alpha$ can be approximated by:
\begin{equation}
\label{eq:alpha}
\begin{array}{l@{\quad}l}
 \alpha
\approx \sum_{i=1}^{m} \sum_{y=0}^{L-1} \max\{\frac{\beta}{Z_{A,\mathcal{B},t}} g_b(A(x_i))_y - \frac{1}{\kappa} g_P(A(x_i))_y, 0\}.
\end{array}
\end{equation}

In Eq~\ref{eq:alpha}, $\beta$ is chosen by the adversary. Thus, we assume that $\beta$ is known, when using $\alpha$ to analyze different backdoor attacks.
Different from $\beta$, $\kappa$ is determined by the trigger function $A$ that distinguishes different backdoor attacks from each other. Next, we demonstrate how to estimate $\kappa$.

\ignore{
\vspace{3pt}\noindent\textbf{Determination of $\beta$}.
Let's consider what's the value of $\beta$ in Eq.~\ref{eq:alpha} should be set.
Lemma~\ref{lemma:xxx} illustrates that the backdoor distance decreases along with the decreasing of $\beta$ and achieves its minimal value when $\beta=\frac{1}{\kappa}$. Besides, $\beta$ is chosen by the adversary and a rational adversary has intend to reduce the backdoor distance to as small as possible to make the injected backdoor be stealthy. Thus, we assume $\beta=\frac{1}{\kappa}$ when comparing with different backdoor attacks to avoid the influences from irrational adversaries. Particularly, when $\beta=\frac{1}{\kappa}$, we have $\alpha = \frac{S}{\kappa}$ where $S$ is defined in Eq.~\ref{eq:S}. Thus, $\alpha$ can be approximated through:
\begin{equation}
\label{eq:alpha_no_beta}
\begin{array}{l@{\quad}l}
 \alpha
\approx \frac{1}{\kappa} \sum_{i=1}^{m} \sum_{y=0}^{L-1} \max\{g_b(A(x_i))_y - g_P(A(x_i))_y, 0\}.
\end{array}
\end{equation}
}

\vspace{3pt}\noindent\textbf{Estimation of $\kappa$}.
Recall that $\kappa = \frac{\Pr(\mathcal{B})}{\Pr(A(\mathcal{B}))}$. Through trivial transformations, we get that
$\kappa = \frac{V(\mathcal{B})}{V(A(\mathcal{B}))} \frac{\mathbb{E}_{\Pr(x | x\in\mathcal{B}) }\Pr(x) }{\mathbb{E}_{\Pr(x | x\in A(\mathcal{B}))}\Pr(x)}$, where $V(\mathcal{B})$ and $V(A(\mathcal{B}))$ are the volumes of set $\mathcal{B}$ and $A(\mathcal{B})$ respectively.
Below, we demonstrate how to estimate $\Pr(x)$ and the volume ratio $\kappa_V=\frac{V(\mathcal{B})}{V(A(\mathcal{B}))}$ separately.

To estimate the prior probability of an input $x$ for the primary task, $\Pr(x)$, we employed a Generative Adversarial Network (GAN)~\cite{stylegan2-ada} and the GAN inversion~\cite{gan-inversion} algorithms. 
Specifically, we aim to build a generative network $G$ and a discriminator network $D$ using adversarial learning: the discriminator $D$ attempts to distinguish the outputs of $G$ and the inputs (e.g., the training samples) $x$ of the primary task, while $G$ takes as the input $z$ randomly drawn from a Gaussian distribution with the variance matrix $I$, i.e., $z \sim \mathbb{N}(0,I)$ and attempts to generate the outputs that cannot be distinguished by $D$. When the adversarial learning converges, the output of $G$ approximately follows the prior probability distribution of $x$, i.e., $\Pr(x) \approx \Pr(G(z)=x))$. In addition, we incorporated with a GAN inversion algorithm capable of recovering the input $z$ of $G$ from a given $x$, s.t., $G(z)=x$. 
Combining the GAN and the inversion algorithm, we can estimate $\Pr(x)$ for a given $x$: 
we first compute $z$ from $x$ using the GAN inversion algorithm, and then estimate $\Pr(x)$ using $\Pr_{\mathbb{N}(0,I)}(z)$.

To estimate the volume ratio $\kappa_V$, we use a Monte Carlo algorithm similar to that proposed by the prior work~\cite{DBLP:journals/corr/abs-2007-06808}. Briefly speaking, for estimating $V(\mathcal{B})$, we first randomly select an $x$ in the backdoor region $\mathcal{B}$ as the origin, and then uniformly sample many directions from the origin and approximate the extent (how long from the origin to the boundary of $\mathcal{B}$) along these directions, and finally, calculate the expectation of the extents of these directions as $Ext(\mathcal{B})$. According to the prior work~\cite{DBLP:journals/corr/abs-2007-06808}, $V(\mathcal{B})$ is approximately equal to the product of $Ext(\mathcal{B})$ and the volume of the $n$ dimensional unit sphere, assuming $\mathcal{B} \subset \mathbb{R}^n$. Therefore, we estimate $\kappa_V$ by $\frac{Ext(\mathcal{B})}{Ext(A(\mathcal{B}))}$.

In general, we estimate $\kappa$ as
$\frac{Ext(\mathcal{B})}{Ext(A(\mathcal{B}))} \frac{\mathbb{E}_{\Pr(x | x\in\mathcal{B}) }\Pr(G^{-1}(x)) }{\mathbb{E}_{\Pr(x | x\in A(\mathcal{B}))}\Pr(G^{-1}(x))}$, where $G^{-1}(x)$ represents the output of a GAN inversion algorithm for a given $x$. We defer the details to Appendix~\ref{app:details_kappa}.

\begin{figure}[ht]
     \centering
     \includegraphics[width=\linewidth]{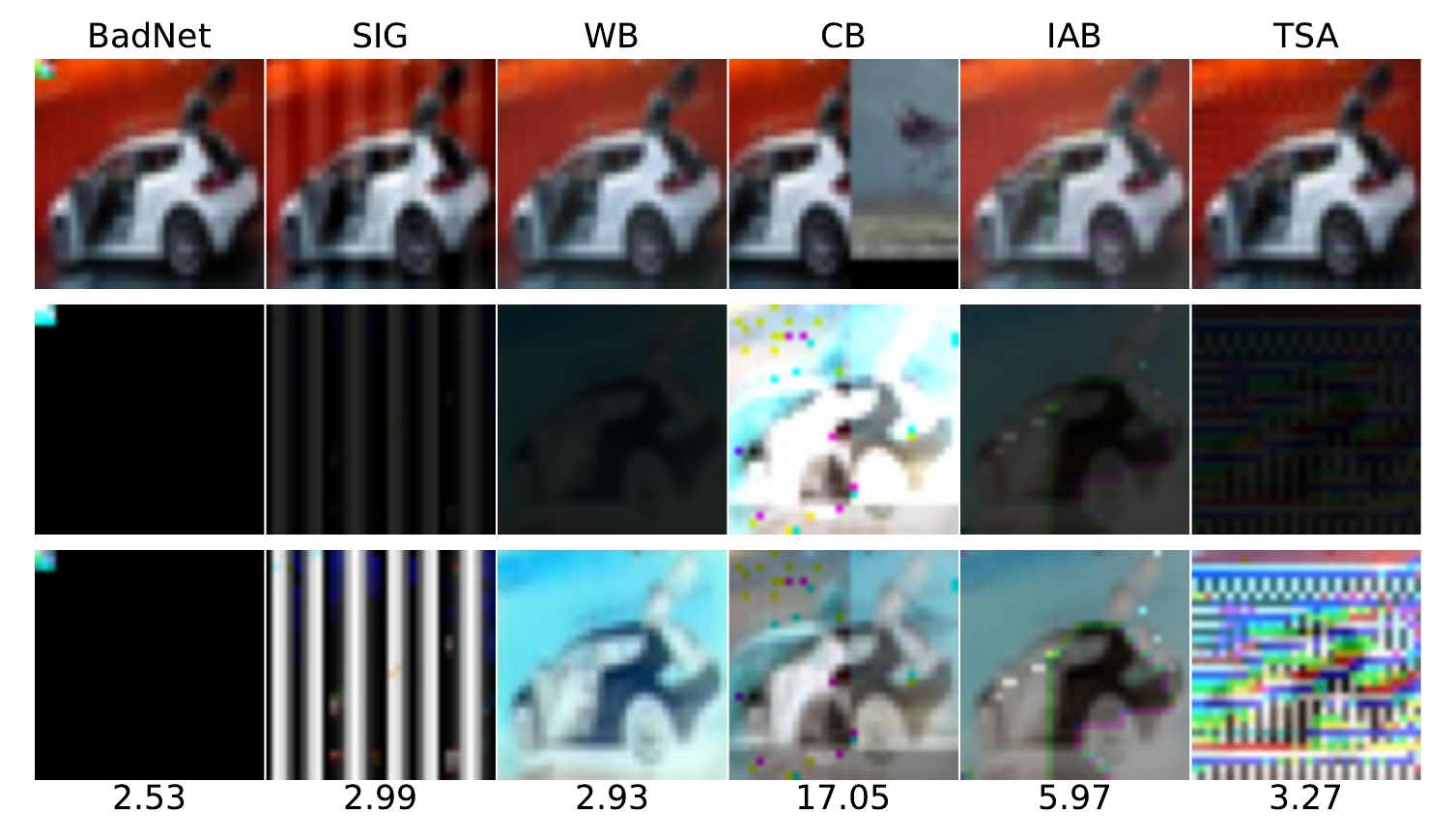}
     \vspace{-15pt}
     \caption{Demonstration of trigger-carrying inputs generated by different attacks. The first row shows attacks' name, the second row presents trigger-carrying inputs, the third row shows triggers, the fourth row shows amplified triggers and the fifth row illustrates the $L_2$-norm of triggers. }
     \label{fig:attack_demo}
\vspace{-10pt}
\end{figure}

\ignore{

Dynamic <<Dynamic backdoor attacks against machine learning models>> EuroS&P'22
[done] LIRA <<LIRA: learnable, Imperceptible and Robust backdoor attacks>> ICCV'21
[done] WaNet <<WaNet: imperceptible warping-based backdoor attack>> ICLR'21
ISSBA <<invisible backdoor attack with sample specific triggers>> ICCV'21
Blind <<Blind Backdoors in Deep Learning Models>> Usenix'21
[done] Wasserstein Backdoor <<Backdoor attack with imperceptible input and latent modification>> NIPS'21                      
[done] RefFool <<Reflection backdoor: a natural backdoor attack on deep neural networks>> ECCV'20                             
[done] Composite Backdoor  <<Composite backdoor attack for deep neural network by mixing existing benign features>> CCS'20
[done] Input-aware backdoor IAB <<Input-Aware Dynamic Backdoor attack>> NIPS'20
[done] Latent backdoor <<latent backdoor attacks on deep neural netowkrs>> ccs'19
[done] SIG <<a new backdoor attack in cnns by training set corruption without label poisoning>> ICIP'19
[done] Label consistent backdoor <<Label-consistent backdoor attacks>> 2019                                                    
[done] TrojanNN <<trojanining attack on neural networks>> NDSS'18
[done] BadNet <<Badnets: Identifying vulnerabilities in the machine learning model supply chain>> 2017 
}

\begin{table}[tbh]
  \centering
  \caption{Backdoor similarities of backdoor attacks. ASR stands for attack success rate, $L_2$-norm stands for the average of $\{\|x-A(x)\|_2 : x \in \mathcal{B} \}$ after regularizing all $x$ and $A(x)$ into $[0,1]^n$. Note that, when $\beta < \frac{1}{\kappa}$, the ``$\alpha/\beta$'' columns represent the $S$ value (Eq.~\ref{eq:S}).}
  \begin{adjustbox}{width=0.48\textwidth}
  
\begin{tabular}{|c|cccc|cccc|c|c|}
\hline
       & \multicolumn{4}{c|}{ASR ($\%$)}                                                                     & \multicolumn{4}{c|}{$\alpha / \beta$}                                                                & $\ln(\kappa)$ & $L_2$-norm    \\ \hline
$\beta$   & \multicolumn{1}{c|}{0.1}   & \multicolumn{1}{c|}{0.05}  & \multicolumn{1}{c|}{0.01}  & 0.005 & \multicolumn{1}{c|}{0.1}  & \multicolumn{1}{c|}{0.05} & \multicolumn{1}{c|}{0.01} & 0.005 & All   & All   \\ \hline
BadNet~\cite{BadNets} & \multicolumn{1}{c|}{99.97} & \multicolumn{1}{c|}{99.18} & \multicolumn{1}{c|}{69.27} & 38.85 & \multicolumn{1}{c|}{0.98} & \multicolumn{1}{c|}{0.98} & \multicolumn{1}{c|}{0.97} & 0.95  & 5.98  & 2.37  \\ \hline
SIG~\cite{SIG}    & \multicolumn{1}{c|}{98.10} & \multicolumn{1}{c|}{81.62} & \multicolumn{1}{c|}{34.57} & 9.88  & \multicolumn{1}{c|}{0.99} & \multicolumn{1}{c|}{0.98} & \multicolumn{1}{c|}{0.96} & 0.94  & 6.05  & 2.72  \\ \hline
WB~\cite{wasserstein-backdoor}     & \multicolumn{1}{c|}{83.38} & \multicolumn{1}{c|}{72.29} & \multicolumn{1}{c|}{32.97} & 7.69 & \multicolumn{1}{c|}{0.67} & \multicolumn{1}{c|}{0.49} & \multicolumn{1}{c|}{0.23} & 0.18  & 4.01  & 2.86  \\ \hline
CB~\cite{composite-backdoor}     & \multicolumn{1}{c|}{86.69} & \multicolumn{1}{c|}{78.18} & \multicolumn{1}{c|}{55.09} & 45.92  & \multicolumn{1}{c|}{1.00} & \multicolumn{1}{c|}{1.00} & \multicolumn{1}{c|}{1.00} & 1.00  & 17.93 & 17.37 \\ \hline
IAB~\cite{IAB}    & \multicolumn{1}{c|}{98.09} & \multicolumn{1}{c|}{92.58} & \multicolumn{1}{c|}{54.37} & 20.13 & \multicolumn{1}{c|}{1.00} & \multicolumn{1}{c|}{1.00} & \multicolumn{1}{c|}{1.00} & 1.00  & 10.72  & 5.96 \\ \hline
TSA (ours)  & \multicolumn{1}{c|}{99.85} & \multicolumn{1}{c|}{99.21} & \multicolumn{1}{c|}{92.83} & 79.07 & \multicolumn{1}{c|}{0.37} & \multicolumn{1}{c|}{0.34} & \multicolumn{1}{c|}{0.32} & 0.25  & 3.07  & 3.13  \\ \hline
\end{tabular}
\end{adjustbox}
	\label{tb:beta_asr_alpha_kappa}
\vspace{-10pt}
\end{table}

\vspace{-15pt}
\subsection{Analysis on Existing Backdoor Attacks}
\label{subsec:tsa_current_attacks}
\vspace{-5pt}

Existing stealthy backdoor attack methods can be summarized into five categories: visually-unrecognizable backdoors, label-consistent backdoors, latent-space backdoors, benign-feature backdoors and sample-specific backdoors. In this section, we report a backdoor similarity analysis on these backdoor attacks, which is important to understanding their stealthiness, given the positive correlation between backdoor distance and detectability we discovered (Section~\ref{sec:detection}).

\ignore{using the method presented in Section~\ref{subsec:alpha_backdoor},}
We compare the backdoor distance of backdoored models generated by 5 different attacks, each representing a different category, on CIFAR10~\cite{cifar10}. 
As mentioned earlier (Theorem~\ref{thr:backdoor_distance_range}), the backdoor distance described by $\alpha$ is related to $\beta$, $\kappa$ and $S$, where $\beta$ is proportional to the \textit{poisoning rate} (see Definition~\ref{def:backdoor_task}), that
describes the adversary's aspiration about how likely those trigger-carrying inputs present in the backdoor distribution in comparison with the probability of showing their benign counterparts in the primary distribution, $\kappa$ also measures the difference between the probability of showing those trigger-carrying inputs and their benign counterparts however within the primary distribution, and $S$ summarizes the conditional probability gain of the outputs given those trigger-carrying inputs obtained on the backdoor distribution compared with such conditional probability on the primary distribution. In simple words, $\beta$ and $\kappa$ together characterizes the difference in inputs and $S$ characterizes the difference in outputs between 
the primary and backdoor distributions.



Specifically, for each attack method, we generated source-specific backdoors (source class is 1 and target class is 0) following the settings described in its original paper but changing $\beta$ to adjust the poisoning rate. In particular,
for $\beta=0.1$, we injected $500$ poisoning samples into the source class (i.e., class 1) with a total of 5,000 samples in the training set.

For each backdoored model, we calculated its ASR at different $\beta$ values as illustrated in Table~\ref{tb:beta_asr_alpha_kappa} to demonstrate the side effect of reducing $\beta$ on ASR. As we see from the table, for BadNet, the ASR is $99.97\%$ when $\beta$ is $0.1$, which goes down with the decrease of the $\beta$, until $38.85\%$ when the $\beta$ drops to $0.005$, rendering the backdoor attack less meaningful. This also shows the rationale of keeping $\beta\geq {1\over\kappa}$, as required in Theorem~\ref{thr:backdoor_distance_range} (here $\beta=0.005 \approx \frac{2}{\kappa}$).    



As illustrated in Eq.~\ref{eq:alpha}, $\alpha$ is proportional to $\beta$, but has a more complicated relation with $\kappa$ and $S$ as further demonstrated in Theorem~\ref{thr:backdoor_distance_range}. To study this complicated relation between $\alpha$ and the parameters other then $\beta$, we normalize $\alpha$ by dividing it with $\beta$ and present the results in Table~\ref{tb:beta_asr_alpha_kappa}. Next, we elaborate our analysis about how existing backdoor attacks reduce $\alpha$ through controlling these parameters.


\vspace{3pt}\noindent\textbf{Visually-unrecognizable backdoors (BadNet)}.
This kind of backdoor attacks generate trigger-carrying inputs visually similar to their benign counterparts, in an attempt to evade the human inspection for anomalous input patterns. 
Generally, visually-unrecognizable backdoors constrain the $L_p$-norm of the trigger, i.e., $\|A(x)-x\|_p$,  to be smaller than a threshold.

Essentially, reducing $\|A(x)-x\|_p$ is to reduce $|\Pr(x) - \Pr(A(x))|$, the difference between the probability of presenting a trigger-carrying input and the probability of presenting its benign counterpart. 
This is because $|\Pr(x) - \Pr(x+\delta)| \propto \|\delta\|_p$, when the perturbation $\delta$ is small and the prior distribution of inputs is some kind of smooth. 
Recall that $\kappa = \Pr(\mathcal{B}) / \Pr(A(\mathcal{B}))$, thus reducing $\|A(x)-x\|_p$ can reduce $\kappa$, as demonstrated in the last two columns of Table~\ref{tb:beta_asr_alpha_kappa}.
However, making $\kappa$ small alone cannot effectively reduce the $\alpha$ as demonstrated by Corollary~\ref{coro:kappa_effect}.
Thus, visually-unrecognizable backdoors only marginally reduce $\alpha$ and moderately increase the backdoor similarity, as observed by our analysis on BadNet (Table~\ref{tb:beta_asr_alpha_kappa}), which only lowers down $\alpha/\beta$ (the normalized $\alpha$) by $0.05$ to $0.95$ when $\beta=0.005$. 



\vspace{3pt}\noindent\textbf{Label-consistent backdoor (SIG)}.
The label-consistent backdoor attacks inject a backdoor into the victim model 
with only 
label-consistent inputs generated by pasting the trigger onto the vague (i.e., hard to be classified) inputs, in an attempt to increase the stealthiness against human inspection. Specifically, prior research~\cite{label-consistent} proposes to use GAN or adversarial examples to get hard-to-classify inputs, while SIG~\cite{SIG} 
utilizes a more inconspicuous trigger (small waves).

However, we found that label-consistent backdoors do not reduce $\alpha$ more effectively, than the naive label-flipped backdoors (e.g., BadNet), 
because injecting a backdoor through label-consistent way has changed neither $\kappa$ nor $S$ of this backdoor task away from that of injecting this backdoor through label-flipped way, as observed in our experiments where similar $\alpha/\beta$ (the normalized $\alpha$) exhibited by these two types of backdoors (see the ``SIG'' and the ``BadNet'' rows in Table~\ref{tb:beta_asr_alpha_kappa}).
Specifically, the BadNet and SIG attacks accomplished their backdoor tasks using similar triggers in terms of $L_2$-norm:
BadNet uses a trigger with the $L_2$-norm of $2.37$ and SIG utilizes a trigger with $L_2$-norm of $2.72$. Apparently, the $\alpha/\beta$ values for the SIG and those for the BadNet are similar at all $\beta$ values we tested. 

\vspace{3pt}\noindent\textbf{Latent-space backdoors (WB)}.
The latent-space backdoor attacks aim to make the backdoored model produce similar latent features for trigger-carrying inputs and benign inputs. 
Prior research~\cite{latent-backdoor} proposes to use this idea to generate a student model that learns the backdoor injected in the teacher model under the transfer learning scenario. Later, this idea has been employed by the Wasserstein Backdoor (WB)~\cite{wasserstein-backdoor} to increase the backdoor stealthiness against the latent space defense (e.g., AC~\cite{AC}). Specifically, WB makes the distribution of the penultimate layer's outputs (latent features) of trigger-carrying inputs as close to those of benign inputs as possible in terms of the sliced-Wasserstein distance~\cite{sliced-wasserstein}.

Making latent features of trigger-carrying inputs and benign inputs be close is essentially to reduce $S$ (defined in Eq.~\ref{eq:S}), the expectation of the conditional probability gain obtained by the backdoored model on trigger-carrying inputs, which is actually the lower-bound of the $\alpha$ when $\beta=\frac{1}{\kappa}$ (Corollary~\ref{coro:kappa_effect}).
In this way, WB effectively reduces $\alpha/\beta$ (the normalized $\alpha$) compared with other four types of backdoors as demonstrated in the ``WB'' row of Table~\ref{tb:beta_asr_alpha_kappa}.
However, this $\alpha/\beta$ reduction achieved by WB comes with the cost of low ASRs, especially when $\beta$ is low (ASR is only $7.69\%$ when $\beta=0.005$), which indicates that reducing $S$ may make the trigger harder to learn.



\vspace{3pt}\noindent\textbf{Benign-feature backdoors (CB)}
Benign-feature backdoor attacks aim to produce backdoored models that leverage features similar to those used by a benign model by constructing a trigger with a composite of benign features, thereby increasing the stealthiness of the backdoor against the backdoor detection techniques that distinguish the weights of backdoored models from those of benign models (e.g., ABS~\cite{ABS}). A representative work in this category is Composite Backdoor (CB)~\cite{composite-backdoor}, which mixes two benign inputs from specific classes into one, and then trains the backdoored model to predict the target labels on these mixed inputs. In another example~\cite{reflection-backdoor}, the adversary constructs the trigger using the reflection features hiding in the input images.

The training inputs with benign features from those in different classes could render the marginal backdoor distribution on inputs significantly deviating from the distribution of benign inputs, making this backdoor even easier to detect. When it comes to backdoor similarity, benign-feature backdoors indirectly reduce $S$ (defined in Eq.~\ref{eq:S}), but in the meantime, increase $\kappa$ (since the trigger-carrying inputs becomes less likely to see from the primary distribution), and thus may not reduce the backdoor distance eventually, which has been shown by the ``CB'' row of Table~\ref{tb:beta_asr_alpha_kappa}.
In addition, the benign-feature backdoors also increase the difficulty in learning the backdoor task (only 83.38\% ASR achieved when $\beta=0.1$).



\vspace{3pt}\noindent\textbf{Sample-specific backdoors (IAB)}
The sample-specific backdoor attacks design the trigger specific to each input. As a result, if an input is given an inappropriate trigger, it will not trigger the backdoor. This kind of backdoors are designed to evade trigger inversion by increasing the difficulty in reconstructing the true trigger.
The Input Aware Backdoor (IAB)~\cite{IAB} is a representative work in this category, which uses a trigger generation network to produce a sample-specific trigger.
The attack methods proposed in the prior work~\cite{ISSBA} and~\cite{dynamic-backdoor} also belong to this category. 

A sample-specific backdoor requires that the trigger carries more information than the trigger of sample-agnostic backdoors, so as to enable the backdoored model to learn the complicated relations between triggers and the inputs. Thus, the trigger of the sample-specific backdoors may come with a large $L_2$-norm. As presented in Table~\ref{tb:beta_asr_alpha_kappa}, the $L_2$-norm of the trigger used by the IAB backdoor is $5.96$, more than twice of the trigger for BadNet ($2.37$) in terms of $L_2$-norm. Such a large trigger renders the trigger-carrying inputs less likely to observe from the primary distribution, thereby reducing the similarity between the probability of seeing benign inputs and the probability of seeing trigger-carrying inputs on the primary distribution, and leading to the increase in $\kappa$ ($\ln(\kappa)=10.72$) and the $\alpha/\beta$ (the normalized $\alpha$).



\vspace{-10pt}
\subsection{New Attack}
\label{subsec:our_attack}
\vspace{-5pt}

In Section~\ref{subsec:tsa_current_attacks} and Table~\ref{tb:beta_asr_alpha_kappa}, we illustrated
that the existing backdoor attacks did not effectively reduce the backdoor distance while keeping high attack success rate (ASR). Our analysis revealed that it is mainly due to three points: 
1) most of these attacks did not reduce $\kappa$ to a small value (e.g., in BadNet, SIG, CB and IAB);
2) the complicated triggers used by many attacks make the backdoor task hard to be learned (e.g., WB, CB and IAB);
and 3) some missed to reduce $S$ (e.g., BadNet, SIG, IAB).

To address these issues, we aim to devise a new attack method that can handle all these points at one time. To reduce $\kappa$, the adversary should use a trigger function that maps a benign input to its close neighbor in terms of not only their $L_p$-norm and but also their probabilities to be presented by the primary task. Using the trigger-carrying inputs with small $L_p$-norm from the benign inputs may not unnecessarily lead to small $\kappa$; in fact, as shown in Table~\ref{tb:beta_asr_alpha_kappa}, the trigger used by BadNet lead to the trigger-carrying inputs with smaller $L_2$-norms but higher $\kappa$ compared to those by WB. On the other hand, to reduce $S$, the adversary should enable the backdoored model to generate similar conditional probabilities as the benign models of the outputs given those trigger-carrying inputs. Finally, the adversary should use a trigger function that can be easily learned; using a complex trigger function as used by WB lead to the backdoored model with low ASR when the poison rate is low (i.e., $\beta$ is small).


At a first glance, it appears impossible to reduce $\kappa$ and $S$ simultaneously, 
as a perfect benign model produces similar outputs for similar inputs and, thus, it always produces different conditional probabilities from the backdoored model of the outputs given those trigger-carrying inputs.
In practice, however, the benign models may not be perfect (highly robust), which may produce very different outputs even for similar inputs, e.g., the adversarial samples~\cite{szegedy2013intriguing}, making it possible to reduce $\kappa$ and $S$ simultaneously.
Together with the trick to make trigger be easy to learn, we the TSA attack which details are illustrated in Algorithm~\ref{alg:tsa_backdoor}.

\begin{algorithm}[htb]
    \caption{TSA attack.}
    \begin{algorithmic}[1]
        \Require{$D_{tr}$, $\mathcal{B}$, $t$, $\alpha^*$, $\beta$, $epoch_{adj}$, $\delta$, $\zeta$, $\omega$} 
        \Ensure{$A(\cdot)$, $f_b$}
        \State Train a benign model $f_P$ on $D_{tr}$ 
        \State Train $A$ with $\mathcal{L}_{A,\mathcal{B},t}(f_P, \alpha^*, \beta)$ (Eq.~\ref{eq:L_A_B_t}) and $\delta$ constraint
        \For{$\_$ \textbf{in} \text{range($epoch_{adj}$)}}
        \State Train $C$ with $\mathcal{L}_{A}(C)$ (Eq.~\ref{eq:L_A})
        \State Update $A$ with $\mathcal{L}_{C}(A, \zeta, \omega)$ (Eq.~\ref{eq:L_d})
        
        \EndFor

        \State   Train $f_b$ on $D_{tr}$ to minimize Eq.~\ref{eq:L_backdoor_train}

    \end{algorithmic}
    \label{alg:tsa_backdoor}
\end{algorithm}

First, in line-$1$, we train a benign model $f_P$ on a given training set $D_{tr}$.
In line-$2$, for the benign model $f_P$, we optimize trigger function $A$ to minimize $\mathcal{L}_{A,\mathcal{B},t}(f_P, -\alpha^*, \beta)$ such that $\|A(x)-x\|_2 \leq \delta$, where
\vspace{-5pt}
\begin{equation}
\label{eq:L_A_B_t}
\begin{array}{l@{\quad}l}
\mathcal{L}_{A,\mathcal{B},t}(f, \alpha^*, \beta) = \mathbb{E}_{(x,y) \in \mathcal{X} \times \mathcal{Y}} \mathcal{L}_{ce}(f(x),y) - \\
\beta \mathbb{E}_{(x,y) \in \mathcal{\mathcal{B}} \times \mathcal{Y}}  (\frac{1+\alpha^*}{2} \log(g(A(x))_t) + \frac{1-\alpha^*}{2} \log(g(A(x))_y)).
\end{array}
\end{equation}
\noindent Here, we assume $f=c \circ g$ as described in Section~\ref{subsec:nn_modeling}. The loss function $\mathcal{L}_{A,\mathcal{B},t}$ is the sum of the loss for the primary task of $f$ on the clean inputs and the loss for the backdoor task of $f$ on the trigger-carrying inputs weighted by $\beta$. The initial trigger function $A$ is an optimized variable, which is trained to minimize $\mathcal{L}_{A,\mathcal{B},t}(f_P, -\alpha^*, \beta)$ while satisfying the $\delta$ constraint, such that this initial trigger function maps the benign inputs to the trigger-carrying inputs in the region of the same class label but close to the classification boundary in $f_P$.
Then in line-$3$ to -$6$, we iteratively refine the trigger function and to make the backdoor task more easily learned. Specifically, we train a small classification network $C$ to distinguish the trigger-carrying inputs from their benign counterparts by minimizing the loss function:
\vspace{-8pt}
\begin{equation}
\label{eq:L_A}
\begin{array}{l@{\quad}l}
\mathcal{L}_A(C) = - \mathbb{E}_{x \in \mathcal{B}} \log(C(A(x)))  + \log(1-C(x)).
\end{array}
\vspace{-5pt}
\end{equation}
\noindent The poor performance of $C$ 
(i.e., $\mathcal{L}_A(C) > \zeta$) indicates that the current trigger function $A$ is hard to learn, and then $A$ is refined to minimize the loss function (line-$5$):
\vspace{-8pt}
\begin{equation}
\label{eq:L_d}
\begin{array}{l@{\quad}l}
\mathcal{L}_C(A, \zeta, \omega) = \mathcal{L}_{A,\mathcal{B},t}(f_P, -\alpha^*, \beta) + \omega \max\{\mathcal{L}_A(C) - \zeta, 0 \},
\end{array}
\vspace{-2pt}
\end{equation}
\noindent which searches for the trigger function $A$ that maps the benign inputs to the trigger-carrying inputs close to the classification boundary in $f_P$ while penalizing those functions $A$ with $\mathcal{L}_A(C) > \zeta$ by incorporating the penalty term with the weight $\omega$.
Finally, in line-$7$, we use the refined trigger function $A$ to poison the training data, which is then used to train a backdoored model $f_b$ by minimizing $\mathcal{L}_{A,\mathcal{B},t}(f_b, \alpha^*, \beta)$ and a regularization term:
\vspace{-8pt}
\begin{equation}
\label{eq:L_backdoor_train}
\begin{array}{l@{\quad}l}
\mathcal{L}_{A,\mathcal{B},t}(f, \alpha^*, \beta) + \|g_b(x_C)-g_P(x_C)\|_2 \\
\text{where\quad}
x_{C} = \underset{x \in \mathcal{X}/A(\mathcal{B})}{\argmax} \|g_b(x) - g_P(x)\|_2.
\end{array}
\vspace{-5pt}
\end{equation}
Here, the regularization term is designed to seek $f_b$ that minimizes the maximum difference between the outputs of $f_b$ and $f_P$ for the inputs without the trigger.


Empirically, we used a LeNet-5~\cite{lecun1998gradient} network as $C$, and an UNet~\cite{UNet} as the trigger function $A$. Besides, we set $epoch_{adj}=3$, $\delta=0.1$, $\zeta = 0.1$ and $\omega=0.1$. We used an Adam~\cite{adam} optimizer with the learning rate of $1e^{-3}$ to train model weights. We implemented our method based on the PyTorch framework and integrated our code into TrojanZoo~\cite{TrojanZOO}.

In our experiments, we used Algorithm~\ref{alg:tsa_backdoor} to generate backdoored models on the CIFAR10 dataset and demonstrate the results in the last row of Table~\ref{tb:beta_asr_alpha_kappa}. 
We observe that the TSA backdoor not only achieved much better ASR ($79.07\%$) than previous attacks ($\leq 50\%$) even when $\beta$ is as small as $0.005$, but also smaller backdoor distance then other attacks at the meanwhile. 
This could be ascribed to several advantages of our approach. First, the trigger function refinement (line-$3$ to -$6$) helps to derive a trigger function that is easy to learn. Second, the $L_p$-norm constraint lets the TSA backdoor has small $\kappa$, which reduces the lower-bound of the backdoor distance (Corollary~\ref{coro:kappa_effect}), and thus allows for the reduction of the backdoor distance by manipulating $S$ (Eq.~\ref{eq:S}). Furthermore, the TSA backdoor attack manages to control the backdoor distance through manipulating $S$ (line-$7$) for a given $\alpha^*$, which enables the TSA backdoor to achieve small backdoor distance on all $\beta$ values. 
In simple words, TSA backdoor maps the benign inputs to the trigger-carrying inputs close to both the classification boundary (controlled by $\alpha^*$) and the original benign inputs (controlled by $\delta$) through a easy-to-learn trigger.

\vspace{-10pt}
\section{TSA on Backdoor Detection}
\label{sec:detection}
\vspace{-5pt}

In the last section, we show that current backdoor attacks are designed for evading specific backdoor detection methods, and do not effectively reduce the backdoor distance that measures how close the backdoor and the primary tasks of a backdoored model are. We further proposed a new TSA attack to strategically reduce the backdoor distance and create more stealthy backdoors. In this section, we demonstrate that the backdoor distance is closely related to the backdoor \textit{detectability}: the backdoors with small backdoor distance are hard to detect, through both theoretical and experimental analysis. 
First, through theoretical analysis, we demonstrate how the backdoor distance affects the evasiveness of a backdoor from detection methods in each of the three classes (Section~\ref{subsec:attack_detection}): detection on model outputs, detection on model weights and detection on model inputs, respectively. Next, we show that in practice, by reducing the backdoor distance, the detectability of a backdoor indeed becomes lower.

\vspace{-10pt}
\subsection{Detection on Model Outputs}
\label{subsec:theory_detection_outputs}
\vspace{-5pt}

The first class of backdoor detection methods, herein referred to as the detection on model outputs, attempt to capture backdoored models by detecting the difference between the outputs of the backdoored models and the benign models on some inputs.
One kind of methods in this class are those methods based on trigger reversion algorithm, which first reconstruct triggers and then check whether a model exhibits backdoor behaviors in response to these triggers. 
In other words, the objective of these methods is to identify some inputs on which the outputs of the backdoored models and of the benign models are different.
When the difference becomes small, however, these methods often become less effective. For example, 
K-ARM~\cite{NeuralCleanse} failed to detect the TSA backdoor (Section~\ref{subsec:detection_experiments}).
Notably, the backdoored model generated by the TSA attack produce the similar outputs as the benign models on the trigger-carrying inputs, and consequently, K-ARM cannot distinguish these two types of models based on the 
reconstructed trigger candidates, even though the $L_p$-norm of those triggers injected by the TSA attack is as small as desired by K-ARM (K-ARM is designed for detecting triggers smaller than a given maximum size). MNTD~\cite{MNTD} is another method in this class, which searches for the inputs on which the backdoored models and the benign models generate the most different outputs.

Formally, we consider the goal of detections in this class is to check whether a trigger function $A(\cdot)$ can be found that maps the inputs to a region $A(\mathcal{B})$, where the outputs of the backdoored model $f_b$ is most different from the outputs of a benign model $f_P$ (e.g., $g_b(A(x))_t \gg g_P(A(x))_t$ for the target label $t$). This goal becomes hard to achieve when the backdoor distance is small, as demonstrated in the following lemma.

\begin{lemma}\label{lemma:prediction_difference}
When $A$ is fixed and $\beta=\frac{1}{\kappa}$, for a well-trained backdoored model $f_b=c_b \circ g_b$ and a well-trained benign model $f_P=c_P \circ g_P$, s.t., $g_b(x)_y \approx \Pr_{A,\mathcal{B},t}(y|x)$ and $g_P(x)_y \approx \Pr(y|x)$ for all $(x,y) \in \mathcal{X}\times \mathcal{Y}$, we have
\vspace{-8pt}
\begin{equation}
\notag
\begin{array}{l@{\quad}l}
\mathbb{E}_{\Pr_{A,\mathcal{B},t}(A(x) | x \in \mathcal{B})} g_b(A(x))_t - \mathbb{E}_{\Pr(A(x) | x \in \mathcal{B})} g_P(A(x))_t \leq \alpha \kappa
\end{array}
\vspace{-5pt}
\end{equation}
\end{lemma}
\begin{proof}
Using Corollary~\ref{coro:beta_effect}, one can derive the desired.
\vspace{-8pt}
\end{proof}

Specifically, Lemma~\ref{lemma:prediction_difference} demonstrates that, when the adversary has chosen a trigger function $A$ and set $\beta=\frac{1}{\kappa}$, which minimizes the backdoor distance in reasonable settings, $\alpha$ is proportional to the upper bound of the difference between the expected outputs of a backdoored model and that of a benign model about the probability of a trigger-carrying input in the target class. 
In other words, when the backdoor distance is small, the difference of the outputs between a backdoored model and a benign model becomes small as well. 
Considering the randomness involved in the training process, when this difference is small, it is hard to distinguish backdoored models from benign models. Therefore, these approaches of detection on model outputs often suffer from false positives, and thus achieve low detection accuracy on the backdoors with small backdoor distance.

\vspace{-10pt}
\subsection{Detection on Model Weights}
\label{subsec:theory_detection_weights}
\vspace{-5pt}

The second class of detection approaches, herein referred to as the detection on model weights, attempt to detect a backdoored model through distinguishing its model weights from those of benign models.
Formally, we consider the goal of detection methods in this class as to verify whether the minimum distance between the weights of a candidate backdoored model $\omega_b$ and the weights of a benign model in a set $\{\omega_P\}$ exceeds a pre-determined threshold $\theta_{\omega}$, i.e., whether $\min_{\omega \in \{\omega_P\}}\|\omega-\omega_b\|_2 > \theta_\omega$. 



To study the difference between the weights of two models, we formulate it as the \textit{weight evolution problem} in continual learning~\cite{thrun1995lifelong}. 
Specifically, we consider two tasks, $\mathcal{T}_P$ and $\mathcal{T}_{A,\mathcal{B},t}$, for which the benign model
$f_P=f(\cdot: \omega_P)$ 
with the weights $\omega_P$
and the backdoored model $f_b=f(\cdot: \omega_b)$ with the weights $\omega_b$ learn to accomplish, respectively. 
We then analyze the change of $\omega_P \to \omega_b$ through the continual learning process $\mathcal{T}_P \to \mathcal{T}_{A,\mathcal{B},t}$.
Based on the Neural Tangent Kernel (NTK)~\cite{NTK} theory, existing work~\cite{lee2019wide} has showed that,
$f_b(x) = f_P(x)+<\phi(x), \omega_b-\omega_P>$ where $\phi(x)$ is the kernel function and $\phi(x) = \bigtriangledown_{\omega_0}f(x; \omega_0)$, which is dependent only on some weights $\omega_0$.
Furthermore, recent research~\cite{DBLP:conf/aistats/DoanBMRA21} has shown that $\|\delta^{\mathcal{T}_P \to \mathcal{T}_{A,\mathcal{B},t}}(X)\|_2^2 = \|\phi(X)(\omega_b-\omega_P)\|_2^2$,
where $\delta^{\mathcal{T}_P \to \mathcal{T}_{A,\mathcal{B},t}}(X)$ is the so-called \textit{task drift} from $\mathcal{T}_P$ to $\mathcal{T}_{A,\mathcal{B},t}$, $\|\delta^{\mathcal{T}_P \to \mathcal{T}_{A,\mathcal{B},t}}(X)\|_2^2 :=  \underset{x \in X}{\Sigma} \|f_b(x)-f_P(x)\|_2^2$. 
Based on these results, we connect the distance between $\omega_P$ and $\omega_b$ to the backdoor distance through the following lemma. 


\begin{lemma}~\label{lemma:weight_distance}
When $A$ is fixed and $\beta=\frac{1}{\kappa}$, for a well-trained backdoored model $f_b=f(\cdot : \omega_b)$, and a well-trained benign model $f_P=f(\cdot : \omega_P)$, we have
\vspace{-5pt}
\begin{equation}
\notag
\begin{array}{r@{\quad}l}
 \frac{\kappa \sqrt{m L}}{\|\phi(X)\|_2 } \alpha \leq \|\omega_b-\omega_P\|_2 .
\end{array}
\vspace{-5pt}
\end{equation}
where $X=\{x_1,x_2,...,x_m\}$ is a set of $m$ inputs in $L$ classes and $\phi(\cdot)$ is the kernel function. 
\end{lemma}
\vspace{-5pt}
\begin{proof}
See Appendix~\ref{proof:weight_distance}.
\vspace{-5pt}
\end{proof}

Lemma~\ref{lemma:weight_distance} demonstrates that, when the adversary has chosen a trigger function $A$ and set $\beta=\frac{1}{\kappa}$, which minimizes the backdoor distance in reasonable settings, $\alpha$ is proportional to the lower bound of the distance between the weights $\omega_b$ and $\omega_P$ in term of $L_2$-norm.
In other words, to ensure the weights of the backdoor models $\omega_b$ is close to the weights of the benign models $\omega_P$, which lead to the backdoors more difficult to be detected by the methods of detection on model weights, the adversary should design a backdoor with small backdoor distance. 



\vspace{-10pt}
\subsection{Detection on Model Inputs}
\label{subsec:theory_detection_inputs}
\vspace{-5pt}

The third class of the detection methods, the detection on model inputs, attempt to identify a backdoored model through detecting the difference between inputs on that the backdoored model and benign models generate similar outputs.
An prominent example of this category is SCAn~\cite{SCAn}, which checks whether the inputs predicted by a backdoored model as belonging to the same class can be well separated into two groups (modeled as two distinct distributions), while the inputs predicted by a benign model as belonging to the same class come from a single group (modeled as a single distribution). Similar idea was also exploited by AC~\cite{AC}. 

We formulate this class of methods as a hypothesis test that evaluates whether the two distributions, characterized by two sets $X_P$ and $X_b$, respectively, on which the benign model $f_P$ and the backdoored model $f_b$ share the same prediction, are significantly different, where $X_P = \{f'_P(x_i): i=1,2,...,n_P\}$ and $X_b=\{f'_b(x_i): i=1,2,...,n_b\}$, $f'_P(x)$ and $f'_b(x)$ are the intermediate results of $f_P(x)$ and $f_b(x)$, respectively, for an input $x$.
For instance, $f'_b(x)$ could be the $j$-th layer's outputs of $f_b$ in a multi-layer neural network. 

Without loss of the generality, we adopt a two-sample Hotelling’s T-square test~\cite{hotelling1992generalization} for this hypothesis test, which tests whether the means of two distributions are significantly different. Here, we consider the test statistic $T^2$, which is calculated from the samples drawn from the two distributions, and is then compared with a pre-selected threshold according to a desirable confidence. The smaller $T^2$, the less probable these two distributions are different in terms of their means. The following lemma demonstrates how the test statistic $T^2$ has an upper-bound related to the backdoor distance.
\begin{lemma}~\label{lemma:normal_dhw1}
When $A$ is fixed and $\beta=\frac{1}{\kappa}$, for a well-trained backdoored model $f_b$ and a well-trained benign model $f_P$, 
if $X_b \sim \mathcal{N}(m_b,\Sigma)$ and $X_P \sim \mathcal{N}(m_P,\Sigma)$ and $n_P$ and $n_b$ are sufficiently large, we have
\begin{equation}
\notag
\begin{array}{r@{\quad}l}
 T^2 \leq \lambda_{max} \frac{n_P n_b}{n_P+n_b} \alpha^2.
\end{array}
\vspace{-5pt}
\end{equation}
where $\lambda_{max}$ is the largest eigenvalue of $\Sigma^{-1}$.
\vspace{-5pt}
\end{lemma}
\begin{proof}
See Appendix~\ref{proof:normal_dhw1}.
\vspace{-8pt}
\end{proof}

Lemma~\ref{lemma:normal_dhw1} demonstrates that, when the adversary has chosen a trigger function $A$ and set $\beta=\frac{1}{\kappa}$, which minimizes the backdoor distance in reasonable settings, $\alpha^2$ is proportional to the upper bound of the test statistic $T^2$. This implies, when the backdoor distance is small, it is difficult to distinguish the distribution of $X_b$ from the distribution of $X_P$, resulting in the poor accuracy of detecting backdoor on model inputs.

\vspace{-10pt}
\subsection{Experiments: Detection vs. Attack}
\label{subsec:detection_experiments}
\vspace{-5pt}

To investigate the performance of these three kinds of detection methods against backdoor attacks, we evaluated 6 backdoor detection methods: K-ARM~\cite{K-ARM}, MNTD~\cite{MNTD}, ABS~\cite{ABS}, TND~\cite{TND}, SCAn~\cite{SCAn} and AC~\cite{AC}, to defend the backdoors injected by 6 backdoor attack methods (Section~\ref{subsec:tsa_current_attacks}) on 4 datasets: CIFAR10~\cite{cifar10}, GTSRB~\cite{gtsrb}, ImageNet~\cite{imagenet} and VGGFace2~\cite{vggface2}.
On each dataset, we generated 200 benign models as the control group. 
For each backdoor attack, we used it to generate 200 backdoored models on every dataset. 
Specifically, in each of the backdoored models, a backdoor was injected with a randomly chosen source class and a randomly chosen target class (different from the source class). We fixed the number of poisoning samples to be equal to $10\%$ of the total number of training samples in the source class, i.e., $\beta=0.1$. Under these settings, we trained the backdoored model  that achieved $>80\%$ ASR for all 6 attack methods on all 4 datasets.
In total, we generated 800 benign models and 4800 backdoored models on all 4 datasets.
To evaluate a backdoor detection method on each dataset, we ran it to distinguish 200 benign models (trained on this dataset) from 200 backdoored models generated by each attack method. Overall, we performed a total of 144 ($=4\times6\times6)$ evaluations on all 4 datasets for all 6 detections against all 6 attacks.
To train a model (benign or backdoored), we used the model structure randomly selected from these four: ResNet~\cite{resnet}, VGG16~\cite{vgg16}, ShuffleNet~\cite{shufflenet} and googlenet~\cite{googlenet}. We used the Adam~\cite{adam} optimizer with the learning rate of $1e^{-2}$ until the model converges (e.g., $\sim 50$ epochs on CIFAR10).


\vspace{3pt}\noindent\textbf{Detection on model outputs}.
We tested two representative detection methods in this category: K-ARM and MNTD. 
K-ARM is one of the winning solutions in TrojAI Competition~\cite{trojai}. It could be viewed as an enhanced version of Neural Cleanse (NC). It cooperates with a
reinforcement learning algorithm to efficiently explore many trigger candidates with different norm and different shape using the trigger reversion algorithm (as used in NC), and thus increases the chance to identify the true trigger. As mentioned by the authors of K-ARM~\cite{K-ARM}, it significantly outperforms NC. Hence, here, we evaluated K-ARM instead of NC.  
MNTD is another representative method in this category. It has been taken as the standard detection method in the Trojan Detection Challenge (TDC)~\cite{TDC}, a NeurIPS 2022 competition. Specifically, MNTD detects the backdoored models by finding some inputs on which the outputs of the backdoored model are most different from the outputs of the benign models. 

\begin{table}[tbh]
  \centering
  \caption{The accuracies (\%) of the detection-on-model-outputs methods. $C$-rows stand for results on CIFAR10, $G$-rows stand for results on GTSRB, $I$-rows stand for results on ImageNet and $V$-rows stand for results on VGGFace2.}
  \begin{adjustbox}{width=0.40\textwidth}
\begin{tabular}{|cc|c|c|c|c|c|c|}
\hline
\multicolumn{2}{|c|}{}                           & BadNet & SIG   & WB    & CB    & IAB   & TSA  \\ \hline
\multicolumn{1}{|c|}{\multirow{4}{*}{K-ARM}} & C & 100  & 61.75 & 79.50 & 57.25 & 80.25 & 59.25 \\ \cline{2-8} 
\multicolumn{1}{|c|}{}                       & G & 100  & 63.25 & 82.25 & 60.50 & 79.75 & 62.50 \\ \cline{2-8} 
\multicolumn{1}{|c|}{}                       & I & 95.50  & 56.50 & 75.00 & 53.75 & 75.00 & 57.25 \\ \cline{2-8} 
\multicolumn{1}{|c|}{}                       & V & 96.25  & 59.25 & 76.50 & 67.25 & 80.75 & 64.75 \\ \hline
\multicolumn{1}{|c|}{\multirow{4}{*}{MNTD}}  & C & 100    & 99.75 & 86.00 & 100   & 98.25 & 51.25 \\ \cline{2-8} 
\multicolumn{1}{|c|}{}                       & G & 100    & 99.25 & 85.50 & 99.50 & 99.50 & 52.50 \\ \cline{2-8} 
\multicolumn{1}{|c|}{}                       & I & 97.75  & 98.75 & 84.25 & 97.25 & 97.25 & 53.25 \\ \cline{2-8} 
\multicolumn{1}{|c|}{}                       & V & 98.75  & 99.00 & 85.25 & 98.25 & 98.75 & 54.75 \\ \hline
\end{tabular}

\end{adjustbox}
	\label{tb:detection_acc_model_outputs}
\vspace{-5pt}
\end{table}

Table~\ref{tb:detection_acc_model_outputs} illustrates that K-ARM works poorly on SIG, CB and TSA, three backdoor attacks using widespread triggers that may affect the whole inputs (even with small $L_2$-norm).
MNTD performs well on the backdoor attacks except on TSA, indicating existing attack methods somehow make the outputs of backdoored models are distant from the outputs of the benign models on many inputs. On the other hand, the outputs of the backdoored model generated by TSA are close to the outputs of benign models. This also helps TSA perform well on TDC competition.\footnote{On TDC, the TSA attack reduced the detection AUC of MNTD to 44.37\%, indicating it is hard for MNTD to distinguish the TSA backdoored models from the benign ones. Until submission of this paper, our method is ranked \#1 in the evasive trojans track of TDC.}

\vspace{3pt}\noindent\textbf{Detection on model weights}.
We tested two representative detection methods in this category: ABS and TND. 
ABS proposed that when the backdoor is injected into a model, it also introduces a short cut, through which a trigger-carrying input will be easily predicted as belonging to the target class by the backdoored model. Specifically, this short cut is characterized by some neurons that are intensively activated by the trigger on the input, and then generate dramatic impact on the prediction. To detect this short cut, ABS first labels those neurons whose activation results in abnormally large change in the predicted label of the backdoored model for some inputs, and then, for each labeled neuron, seeks a trigger that can activate this neuron abnormally and consistently change the predicted label for a range of inputs. ABS alarms for a backdoored model, if such neuron coupled with a trigger is found.
TND explores another phenomenon related to the short cut in a model. In particular, TND found that the untargeted universal perturbation is similar to the targeted per-input perturbation in backdoored models, while they are different in benign models. Hence, TND alarms for a backdoored model if such similarity is significant.

\begin{table}[tbh]
\vspace{-5pt}
  \centering
  \caption{The accuracy (\%) of the detection-on-model-weights methods. $C$-rows stand for results on CIFAR10, $G$-rows stand for results on GTSRB, $I$-rows stand for results on ImageNet and $V$-rows stand for results on VGGFace2.}
  \begin{adjustbox}{width=0.40\textwidth}
\begin{tabular}{|cc|c|c|c|c|c|c|}
\hline
\multicolumn{2}{|c|}{}                         & BadNet & SIG   & WB    & CB    & IAB   & TSA  \\ \hline
\multicolumn{1}{|c|}{\multirow{4}{*}{ABS}} & C & 100    & 95.50 & 59.75 & 62.75 & 58.75 & 51.00 \\ \cline{2-8} 
\multicolumn{1}{|c|}{}                     & G & 100    & 94.50 & 61.00 & 61.75 & 59.00 & 49.25 \\ \cline{2-8} 
\multicolumn{1}{|c|}{}                     & I & 94.75  & 89.75 & 56.25 & 58.00 & 55.50 & 51.75 \\ \cline{2-8} 
\multicolumn{1}{|c|}{}                     & V & 98.25  & 91.25 & 56.25 & 59.25 & 54.75 & 52.25 \\ \hline
\multicolumn{1}{|c|}{\multirow{4}{*}{TND}} & C & 100    & 99.75 & 67.00 & 73.75 & 53.00 & 48.75 \\ \cline{2-8} 
\multicolumn{1}{|c|}{}                     & G & 100    & 99.25 & 64.25 & 72.50 & 52.50 & 50.25 \\ \cline{2-8} 
\multicolumn{1}{|c|}{}                     & I & 94.50  & 93.25 & 62.00 & 69.25 & 49.75 & 51.50 \\ \cline{2-8} 
\multicolumn{1}{|c|}{}                     & V & 96.00  & 92.75 & 63.50 & 71.00 & 50.25 & 51.75 \\ \hline
\end{tabular}

\end{adjustbox}
	\label{tb:detection_acc_model_weights}
\vspace{-5pt}
\end{table}

Table~\ref{tb:detection_acc_model_weights} illustrates that ABS suffers from difficulties when defending against WB, CB, IAB and TSA, perhaps because these attacks influence many neurons in the victim models and thus no single neuron changes the predicted label by itself. 
On the other hand, TND performs better than ABS when defending against WB and CB, indicating the similarity between the untargeted universal perturbation and targeted per-input perturbations is a more general signal of the short cut comparing to a single dominant neuron exploited by ABS.

\vspace{3pt}\noindent\textbf{Detection on model inputs}.
We tested two representative detection methods in this category: SCAn and AC.
SCAn detects the backdoor by checking whether the representations (outputs of the penultimate layer) of inputs in a single class are from a mixture of two distributions, with the help of the so-called global variance matrix that captures how the representations of the inputs in different classes varies. SCAn first computes the global variance matrix on a clean dataset, then computes a score for each class, and finally checks whether any class has a abnormally high score. If such class exists, SCAn will report this model as the backdoored model and this abnormal class as the target class. 
Similarly, AC detects the backdoor by checking whether the representations of one class can be well separated into two groups. Specifically, for each class, AC first embeds the high-dimensional representations into 10-dimensional vectors and then computes the Sihouette score~\cite{silhouettes} to measures how well the 2-means algorithm can separate these vectors.

\begin{table}[tbh]
\vspace{-5pt}
  \centering
  \caption{The accuracies (\%) of the detection-on-model-outputs methods. $C$-rows stand for results on CIFAR10, $G$-rows stand for results on GTSRB, $I$-rows stand for results on ImageNet and $V$-rows stand for results on VGGFace2.}
  \begin{adjustbox}{width=0.40\textwidth}
\begin{tabular}{|cc|c|c|c|c|c|c|}
\hline
\multicolumn{2}{|c|}{}                          & BadNet & SIG   & WB    & CB    & IAB   & TSA  \\ \hline
\multicolumn{1}{|c|}{\multirow{4}{*}{SCAn}} & C & 100    & 100   & 70.25 & 95.25 & 74.25 & 63.25 \\ \cline{2-8} 
\multicolumn{1}{|c|}{}                      & G & 100    & 100   & 69.00 & 97.00 & 74.75 & 61.75 \\ \cline{2-8} 
\multicolumn{1}{|c|}{}                      & I & 94.25  & 91.25 & 62.75 & 88.00 & 67.75 & 59.00 \\ \cline{2-8} 
\multicolumn{1}{|c|}{}                      & V & 95.75  & 92.00 & 66.25 & 89.50 & 69.50 & 60.25 \\ \hline
\multicolumn{1}{|c|}{\multirow{4}{*}{AC}}   & C & 98.00  & 99.00 & 59.75 & 90.00 & 65.75 & 55.25 \\ \cline{2-8} 
\multicolumn{1}{|c|}{}                      & G & 98.50  & 99.25 & 59.25 & 91.50 & 66.50 & 55.25 \\ \cline{2-8} 
\multicolumn{1}{|c|}{}                      & I & 91.75  & 95.50 & 55.75 & 86.25 & 59.75 & 52.75 \\ \cline{2-8} 
\multicolumn{1}{|c|}{}                      & V & 92.25  & 96.25 & 57.25 & 88.00 & 62.50 & 53.50 \\ \hline
\end{tabular}

\end{adjustbox}
	\label{tb:detection_acc_model_inputs}
\vspace{-5pt}
\end{table}

Table~\ref{tb:detection_acc_model_inputs} demonstrates that SCAn achieved better accuracies against all 6 attacks compared to AC. However, SCAn and AC both performed poorly on WB, IAB and TSA, the three attacks that attempt to mix the representations of trigger-carrying inputs with those of benign inputs.

Taking all these results 
together, we concluded that an attack would exhibit different evasiveness against different detection methods. Even for TSA, although the detection accuracy by 4 out of 6 detections are as low as about $52\%$, two other methods (K-ARM and SCAn) retain about $60\%$ accuracy against it. 
This illustrates the demand of a general measurement to depict how well a backdoor attack can evade different detection methods (including novel methods that are not known by the adversary), as in practice, the defender may adopt a cocktail approach by combining different methods to detect backdoors. 
We believe the backdoor distance is a promising candidate for such a measurement as it accurately showed the low detection accuracy on the TSA and WB backdoored models by all detection methods, with their low backdoor distances (Table~\ref{tb:beta_asr_alpha_kappa}) compered to the other 4 attack methods. Below, we aim to further illustrate their connection.

\vspace{-10pt}
\subsection{Experiments: Detectability vs. Similarity}
\label{subsec:detectability_experiments}
\vspace{-5pt}
Our experiments in Section~\ref{subsec:detection_experiments} indicate that the backdoor distance is a potentially good measurement of the backdoor detectability (as defined in below). 
Specifically, those backdoor attacks obtaining small backdoor distance are hard to be detected, which is also inline with what has been demonstrated in our theory analysis (Section~\ref{subsec:theory_detection_outputs}~\ref{subsec:theory_detection_weights} and ~\ref{subsec:theory_detection_inputs}). 
In this section, we report the experimental results showing the backdoors with small backdoor distance indeed have low detectability, and thus the backdoor distance is indeed a good indicator of the backdoor detectability.

\begin{definition}[Backdoor detectability]~\label{def:backdoor_detectability}
The detectability of the backdoor generated by a backdoor attack method is the maximum accuracy that backdoor detection methods can achieve to distinguish the backdoored model from the benign models. For convenience, we adjust the detectability between 0 and 1, i.e.,  $\gamma = |acc-0.5| \times 2$, where $\gamma$ is the detectability and $acc$ is the maximum accuracy.
\end{definition}

\vspace{-10pt}
\begin{table}[tbh]
\vspace{-5pt}
  \centering
  \caption{Detectability for attacks. $C$-rows stand for results on CIFAR10, $G$-rows stand for results on GTSRB, $I$-rows stand for results on ImageNet dataset and $V$-rows stand for results on VGGFace2 dataset. The ``Det'' columns represent the backdoor detectability. The ``$\alpha/\beta$'' columns depict the backdoor distance (Corollary~\ref{coro:beta_effect}). Here, we keep $\beta=0.1$ for all cells.}
  \begin{adjustbox}{width=0.48\textwidth}
\begin{tabular}{|c|cc|cc|cc|cc|cc|cc|}
\hline
\multirow{2}{*}{} & \multicolumn{2}{c|}{BadNet}       & \multicolumn{2}{c|}{SIG}          & \multicolumn{2}{c|}{WB}           & \multicolumn{2}{c|}{CB}           & \multicolumn{2}{c|}{IAB}          & \multicolumn{2}{c|}{Ours}         \\ \cline{2-13} 
                  & \multicolumn{1}{c|}{Det}  & $\alpha/\beta$ & \multicolumn{1}{c|}{Det}  & $\alpha/\beta$ & \multicolumn{1}{c|}{Det}  & $\alpha/\beta$ & \multicolumn{1}{c|}{Det}  & $\alpha/\beta$ & \multicolumn{1}{c|}{Det}  & $\alpha/\beta$ & \multicolumn{1}{c|}{Det}  & $\alpha/\beta$ \\ \hline
C                 & \multicolumn{1}{c|}{1.00} & 0.98  & \multicolumn{1}{c|}{1.00} & 0.99  & \multicolumn{1}{c|}{0.72} & 0.67  & \multicolumn{1}{c|}{1.00} & 1.00  & \multicolumn{1}{c|}{0.97} & 1.00  & \multicolumn{1}{c|}{0.27} & 0.37  \\ \hline
G                 & \multicolumn{1}{c|}{1.00} & 0.96  & \multicolumn{1}{c|}{1.00} & 1.00  & \multicolumn{1}{c|}{0.71} & 0.61  & \multicolumn{1}{c|}{0.99} & 1.00  & \multicolumn{1}{c|}{0.99} & 1.00  & \multicolumn{1}{c|}{0.25} & 0.41  \\ \hline
I                 & \multicolumn{1}{c|}{0.96} & 0.92  & \multicolumn{1}{c|}{0.98} & 0.98  & \multicolumn{1}{c|}{0.69} & 0.66  & \multicolumn{1}{c|}{0.96} & 0.99  & \multicolumn{1}{c|}{1.00} & 0.99  & \multicolumn{1}{c|}{0.18} & 0.38  \\ \hline
V                 & \multicolumn{1}{c|}{0.98} & 0.95  & \multicolumn{1}{c|}{0.98} & 0.99  & \multicolumn{1}{c|}{0.71} & 0.65  & \multicolumn{1}{c|}{0.97} & 1.00  & \multicolumn{1}{c|}{0.98} & 1.00  & \multicolumn{1}{c|}{0.30} & 0.35  \\ \hline
\end{tabular}

\end{adjustbox}
	\label{tb:detectability_distance}
\vspace{-10pt}
\end{table}

To evaluate the relationship between the backdoor detectability and the backdoor distance for each attack method, we summarized the maximum detection accuracy obtained among 6 detection methods (Section~\ref{subsec:detection_experiments}) and calculated the detectability according to the above definition. 
Also we approximated the backdoor distance of these 6 attacks on 4 datasets using our approximation method (Section~\ref{subsec:alpha_backdoor}) with the help of StyleGAN2 models~\cite{stylegan2-ada}. Specifically, we used the officially pretrained StyleGAN2 models for datasets CIFAR10 and ImageNet, trained a StyleGAN2 model with its original code for GTSRB dataset, and trained a StyleGAN2 model with code~\cite{stylegan2-vggface2} available online for VGGFace2 dataset.
Our results are illustrated in Table~\ref{tb:detectability_distance}.

From Table~\ref{tb:detectability_distance}, we observe that the backdoor detectability is roughly equal to the backdoor distance (depicted by $\alpha/\beta$). Digitally, the Pearson correlation coefficient~\cite{person_coefficient} between them is $0.9777$, the mean value of the absolute difference between them is $0.0450$ and the standard deviation of that is $0.0498$. These numbers demonstrate that the backdoor distance is highly correlated to and a good indicator of the backdoor detectability, when the $\alpha/\beta$ be close to $1$, be about $0.64$ (WB) or be around $0.37$ (TSA).

To evaluate this relationship at more various backdoor distances, we used the TSA backdoor attack method (Algorithm~\ref{alg:tsa_backdoor}) with $\beta=0.1$ to generate backdoors with different backdoor distances by adjusting the parameter $\alpha^*$. 
Specifically, we performed this experiment on CIFAR10 with 9 different $\alpha^*$ values ranging from 0.1 to 0.9.
For each $\alpha^*$, we generated 200 backdoored models together with previously generated benign models for CIFAR10 to build a testing dataset containing 400 models (200 backdoored and 200 benign models). 
On each testing dataset, we applied SCAn and K-ARM detection methods, two comparably effective detection methods against TSA backdoor attack (Section~\ref{subsec:detection_experiments}), to distinguish those TSA backdoored models from benign models.
Figure~\ref{fig:alpha_detectability} demonstrates these detection results and the backdoor distances we estimated on each $\alpha^*$ value. 
\vspace{-10pt}
\begin{figure}[ht]
     \centering
     \includegraphics[width=\linewidth]{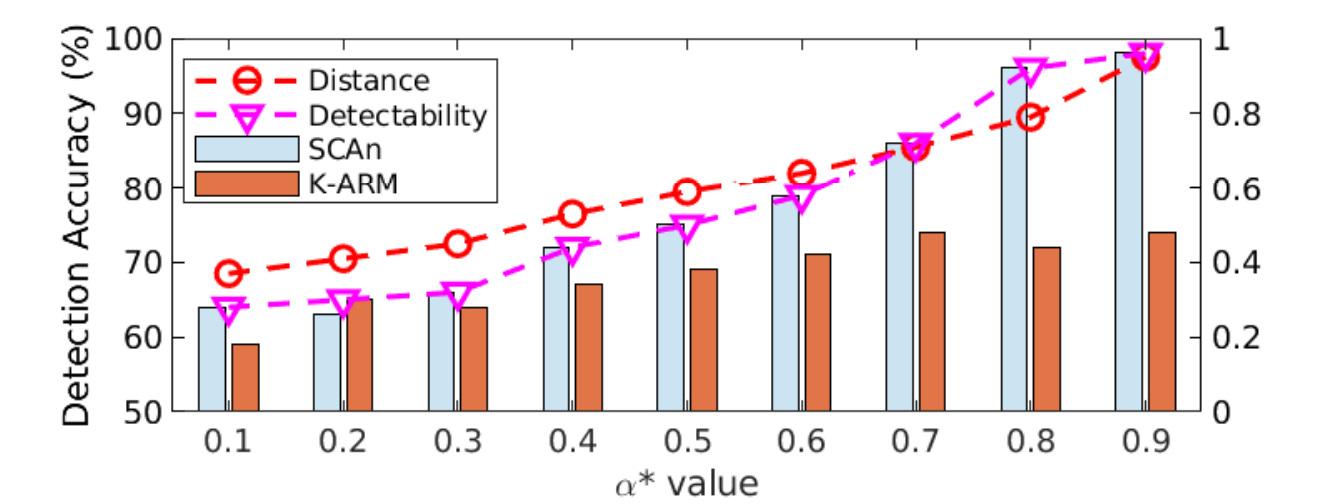}
     \caption{Backdoor distance and detectability of backdoors generated by using different $\alpha^*$ values.}
     \label{fig:alpha_detectability}
 \vspace{-10pt}
\end{figure}

From Figure~\ref{fig:alpha_detectability}, we observe that the backdoor detectability (blue line) of TSA increases along with the increasing backdoor distance (red line, characterized by $\alpha/\beta$). Notably, a small difference exists between these two lines,
because of two reasons: 1) imprecise estimation of the backdoor distance; and 2) the absence of an effective detection method against TSA backdoors.
Specifically, when $\alpha^* \ge 0.8$, the backdoor distance is lower than the backdoor detectability, illustrating the first reason. When $\alpha^* \le 0.5$, the backdoor distance is consistently higher than the backdoor detectability, demonstrating the room for better detection methods (the second reason).

Based upon our results, we conclude that the backdoor distance is a good indicator of the backdoor detectability with small deviations, which is again illustrated by the following observation: the Pearson correlation coefficient between the backdoor distances and detectabilities shown on Figure~\ref{fig:alpha_detectability} is $0.9795$, while the mean  of the absolute difference between them is $0.0800$ with the standard deviation of  $0.0453$.

\vspace{-15pt}
\section{Mitigation}
\label{sec:mitigation}
\vspace{-10pt}

A simple defense to the backdoors with small backdoor distance could be just discarding those uncertain predictions while retaining only those confident predictions. However, doing this will obviously decrease the model accuracy on benign inputs. For example, on MNIST dataset, if keeping only those predictions with confidence higher than $0.8$ and labeling the rest as ``unknown'', the accuracy of a benign model will decrease from $99.35\%$ to $98.61\%$, this is far below the accuracy of a benign model could get\footnote{11 times of standard deviation below the mean accuracy}, considering the mean accuracy among 200 benign models is $99.25\%$ with the standard deviation of $0.00057$. When the primary task becomes more hard (e.g., ImageNet), the accuracy reduction will be more serious if this simple defense be applied.

Besides, backdoor unlearning methods and backdoor disabling methods might have potential to relieve the threat from backdoors with small backdoor distance. However, as demonstrated in our TSA on them (Appendix~\ref{app:backdoor_unlearning} $\&$~\ref{app:backdoor_disabling}), they exhibit minor efficacy on these backdoors. 

Inspired by our backdoor distance theorems, a detection that considers both the difference exhibited in the inputs and in the outputs between the backdoored model and benign models would effectively reduce the evasiveness of those backdoors with small backdoor distance, which is a promising direction to develop powerful detections in futures.

\vspace{-15pt}
\section{Related Works}
\vspace{-10pt}
We proposed theorems to study the detectability of backdoors.
This, in general, has also been studied in previous work~\cite{goldwasser2022planting}.
It proposed an approach to plant undetectable backdoors into a random feature network~\cite{rahimi2007random}, a kind of neural network that learns only the weights on random features. Compared with classical deep neural networks, random feature networks have limited capability~\cite{yehudai2019power}: it cannot be used to learn even a single ReLU neuron, unless the network size is exponentially larger than the dimension of the inputs. 
In theory, work~\cite{goldwasser2022planting} reduced the problem of detecting their backdoor to solving a Continuous Learning With Errors problem~\cite{bruna2021continuous}, however, solving which is as hard as finding approximately short vectors on arbitrary integer lattices, thus detecting their backdoor is computationally infeasible in practice. 
Compared with work~\cite{goldwasser2022planting}, our work established theorems about the detectability of backdoors injected into a classical deep neural network, and demonstrated that, in this case, backdoor detectability is characterized by the backdoor distance that further controlled by three parameters: $\kappa$, $\beta$ and $S$.

Based upon our theorem, we proposed an attack, TSA backdoor attack, to inject stealthy backdoor. Compared to existing stealthy backdoor attacks~\cite{BadNets}~\cite{SIG}~\cite{wasserstein-backdoor}~\cite{composite-backdoor}~\cite{IAB}, TSA backdoor attack achieved lower backdoor detectability under current backdoor detections~\cite{K-ARM}~\cite{MNTD}~\cite{ABS}~\cite{TND}~\cite{SCAn}~\cite{AC} (demonstrated in Section~\ref{subsec:detection_experiments}) and has theory guarantee under unknown detections (illustrated in Section~\ref{subsec:theory_detection_outputs}~\ref{subsec:theory_detection_weights}~\ref{subsec:theory_detection_inputs}).

Our TSA backdoor attack exploited adversarial perturbations as the trigger, that has been also exploited by IMC backdoor attack~\cite{pang2020tale}. There are 3 main differences between TSA and IMC backdoor attacks: 1) TSA has theory guarantees on backdoor detectability what IMC has not; 2) TSA reduces the $S$ (defined in Eq.~\ref{eq:S}) what IMC has not considered; 3) TSA reduces the difference between outputs of backdoored model and the benign model on whole input space (Eq.~\ref{eq:L_backdoor_train}), however, IMC only reduced the difference between outputs on benign inputs (i.e., maintained the accuracy of backdoored model on benign inputs). And, there are two minor differences between them: a) TSA makes the trigger be easy to learn while IMC did not; b) TSA only slightly changes the classification boundary, however, IMC iteratively pushed the classification boundary deviate from its original position to seek a small trigger. 
Also, we established experiments to compare TSA with IMC on CIFAR10 dataset (see Appendix~\ref{app:IMC} for details), in which TSA exhibited lower detectability than IMC.
\vspace{-15pt}
\section{Limitations and Future Works}
\vspace{-10pt}
In this work, we only studied backdoor tasks where $\beta \geq \frac{1}{\kappa}$, i.e., the adversary has not reduced the probability of drawing a trigger-carrying input from the backdoor distribution be lower than the probability of drawing it from the primary distribution. However, as demonstrated in Table~\ref{tb:beta_asr_alpha_kappa}, when $\beta < \frac{1}{\kappa}$, TSA still achieved acceptable ASR (ASR=$79.07\%$ when $\beta = 0.005 < 0.046 = \frac{1}{\kappa}$), illustrating the need to extend our theorem to adapt $\beta < \frac{1}{\kappa}$ scenarios.
However, using the similar methods applied on $\beta \geq \frac{1}{\kappa}$ scenarios, one could easily obtain that the minimal backdoor distance will be obtained at $\beta=\frac{1}{\kappa}$ even in $\beta < \frac{1}{\kappa}$ scenarios, which is inline with the conclusion drew for $\beta \geq \frac{1}{\kappa}$ scenarios (Corollary~\ref{coro:beta_effect} $\&$~\ref{coro:kappa_effect}) and has no conflicts with the results shown in Table~\ref{tb:beta_asr_alpha_kappa}.

In section~\ref{sec:mitigation}, we have only taken the first step to use our backdoor distance theorem to understand the backdoor unlearning and backdoor disabling methods. Comprehensive studies are needed in the future.

Our Theorem~\ref{thr:backdoor_general_distance} reveal that the fundamental difference between a backdoored model and a benign model comes from the difference between their joint probabilities among trigger-carrying inputs and the outputs (i.e., $A(\mathcal{B}) \times \mathcal{Y}$). This implies that a good backdoor detection method should simultaneously consider the differences in the outputs and in the inputs between backdoored models and benign models, rather than considering one of these two differences alone as what current detection methods did. Actually, to detect backdoor, this points out a potential direction for the future studies.

\vspace{-15pt}
\section{Conclusion}
\vspace{-10pt}
We established theorems about the backdoor distance (similarity) and used them to investigate the stealthiness of current backdoors, revealing that they have taken only some of factors affecting the backdoor distance into the consideration. 
Thus, we proposed a new approach, TSA attack, which simultaneously optimizes those factors under the given constraint of backdoor distance.
Through theoretical analysis and extensive experiments, we demonstrated that the backdoors with smaller backdoor distance were in general harder to be detected by existing backdoor defense methods. Furthermore, comparing with existing backdoor attacks, the TSA attack generates backdoors that exhibited smaller backdoor distances, and thus lower detectability under current backdoor detections.

\bibliographystyle{plain}
\bibliography{main}

\newpage

\section{Appendix of Details of Estimation $\kappa$}
\label{app:details_kappa}

\vspace{3pt}\noindent\textbf{Estimation of $\kappa_{\Pr}$}.
In Section~\ref{subsec:alpha_backdoor}, we described how to calculate $\Pr(x)$ for an given input $x$. Specifically, our implementation is based on the GAN inversion tools in the official repository of~\cite{stylegan2-ada}. The original code of~\cite{stylegan2-ada} can only recover inputs' style parameters that actually are projections of the $z$ through a transformation network. Thus, we modified the original code to directly recover $z$.

After computing $\Pr(x)$, it is still not trivial to get the expectations, $\mathbb{E}_{\Pr(x| x \in \mathcal{B})} \Pr(G^{-1}(x))$ and $\mathbb{E}_{\Pr(x| x \in A(\mathcal{B}))} \Pr(G^{-1}(x))$, due to the poor precision in computing tiny numbers (e.g., $1e^{-100}$). Thus, we, instead, compute the logarithm of the ratio, i.e., $\ln(\kappa_{\Pr}) = \ln(\frac{\mathbb{E}_{\Pr(x| x \in \mathcal{B})} \Pr(G^{-1}(x))}{\mathbb{E}_{\Pr(x| x \in A(\mathcal{B}))} \Pr(G^{-1}(x))}$).
Furthermore, we observed that $z_x = G^{-1}(x)$ follows a Gaussian distribution for both $x \in \mathcal{B}$ and $x \in A(\mathcal{B})$. Combining with the fact that $\Pr(z_x) \propto \|z_x\|^2_2$, we get that
\begin{equation}
\notag
\begin{array}{l@{\quad}l}
\ln(\kappa_{\Pr}) = -\frac{1}{2} (\frac{\mu_{\mathcal{B}}^2}{\sigma_{\mathcal{B}}^2+1} - \frac{\mu_{A(\mathcal{B})}^2}{\sigma_{A(\mathcal{B})}^2+1}),
\end{array}
\end{equation}
\noindent where we assume $z_x \sim \mathbb{N}(\mu_{\mathcal{B}}, \sigma_{\mathcal{B}})$ for $x \in \mathcal{B}$ and, for $x \in A(\mathcal{B})$, $z_x \sim \mathbb{N}(\mu_{A(\mathcal{B})}, \sigma_{A(\mathcal{B}}))$.
Empirically, we sampled 100 points ($x$) in $\mathcal{B}$ and 100 points ($x$) in $A(\mathcal{B})$ to estimate the mean and the variance of the corresponding $z_x$.

\vspace{3pt}\noindent\textbf{Estimation of $\kappa_V$}.
In Section~\ref{subsec:alpha_backdoor}, we described $\kappa_V = \frac{Ext(\mathcal{B})}{Ext(A(\mathcal{B}))}$, and briefly introduced how to calculate $Ext(\mathcal{B})$.
Specifically,  for a randomly selected origin $x \in \mathcal{B}$, we sampled a set of $256$ other inputs $\{x_1, x_2, ..., x_{256}\}$.
We then generated a set of $256$ random directions from $x$ as $\{\frac{x_1-x}{\|x_1-x\|_2}, \frac{x_2-x}{\|x_2-x\|_2}, ..., \frac{x_{256}-x}{\|x_{256}-x\|_2}\}$, and along each direction, we used an binary search algorithm to find the extent (i.e., how far the origin $x$ is from the boundary). 
Take a source-specific backdoor (with the source class of $1$ and the target class of $0$) as an example. We used a benign model $f_P$ and a backdoored model $f_b$ to detect the boundary:
$f_P(x) = 1$ and $f_b(x) = 1 \Leftrightarrow x \in \mathcal{B}$; $f_P(x) = 1$ and $f_b(x) = 0 \Leftrightarrow x \in A(\mathcal{B})$.
Finally, for computing $Ext(\mathcal{B})$, we randomly selected $32$ different origins, and set $Ext(\mathcal{B})$ as the average extent among those computed from these origins. The same method was also used to compute $Ext(A(\mathcal{B}))$.

\section{Appendix of Proofs}

\subsection{Proof of Proposition~\ref{prop:h-w1}}
\label{proof:h-w1}

The inequality $0 \leq d_{\mathcal{H}-W1}(\mathcal{D},\mathcal{D'}) \leq 1$ is obvious, and thus we omit the proof. Next, we focus on proving $d_{W1}(\mathcal{D},\mathcal{D'}) \leq d_{\mathcal{H}-W1}(\mathcal{D},\mathcal{D'})$ and $ d_{\mathcal{H}-W1}(\mathcal{D},\mathcal{D'}) = \frac{1}{2} d_{\mathcal{H}}(\mathcal{D},\mathcal{D'})$.

Let's recall the definitions of \textit{Wasserstein-1 distance} and \textit{$\mathcal{H}$-divergence}:

\vspace{2pt}\noindent$\bullet$~\textit{Wasserstein-1 distance:}
Assuming $\Pi(\mathcal{D}, \mathcal{D}')$ is the set of joint distributions $\gamma(x,x')$ whose marginals are $\mathcal{D}$ and $\mathcal{D}'$, respectively, the Wasserstein-1 distance between $\mathcal{D}$ and $\mathcal{D}'$ is
\begin{equation}
\begin{array}{r@{\quad}l}
d_{W1}(\mathcal{D}, \mathcal{D}') = \underset{\gamma \in \Pi(\mathcal{D}, \mathcal{D}')}{\inf} \mathbb{E}_{(x,x') \sim \gamma} \|x-x'\|_2.
\end{array}
\label{eq:w1-distance}
\end{equation}

\vspace{2pt}\noindent$\bullet$~\textit{$\mathcal{H}$-divergence:}
Given two probability
distributions $\mathcal{D}$ and $\mathcal{D}'$ over the same domain $\mathcal{X}$, we consider a hypothetical binary classification on $\mathcal{X}$: $\mathcal{H} = \{h:\mathcal{X} \mapsto \{0,1\}\}$, and denote by $I(h)$ the set for which $h \in \mathcal{H}$ is the characteristic function, i.e., $x \in I(h) \Leftrightarrow h(x) = 1$. The $\mathcal{H}$-divergence between $\mathcal{D}$ and $\mathcal{D}'$ is
\begin{equation}
\begin{array}{r@{\quad}l}
d_{\mathcal{H}}(\mathcal{D}, \mathcal{D}') = 2 \underset{h \in \mathcal{H}}{\sup} |\Pr_{\mathcal{D}}(I(h)) - \Pr_{\mathcal{D}'}(I(h)) |.
\end{array}
\label{eq:h-divergence}
\end{equation}

From the definition of \textit{$\mathcal{H}$-divergence} (Eq.~\ref{eq:h-divergence}), one can directly obtain that $ d_{\mathcal{H}-W1}(\mathcal{D},\mathcal{D'}) = \frac{1}{2} d_{\mathcal{H}}(\mathcal{D},\mathcal{D'})$, using the fact that $\max(\{0,1\}) = \max([0,1]) = 1$ and $\min(\{0,1\}) = \min([0,1]) = 0$.

Referring to the prior work~\cite{w-distance}, we apply the Kantorovich-Rubinstein duality~\cite{villani2009optimal} to transform Eq.~\ref{eq:w1-distance} into its dual form:
\begin{equation}
\begin{array}{r@{\quad}l}
\hat{d}_{W1}(D, D') = \underset{\|h\|_{\mathcal{L}} \leq 1}{\max} [ \mathbb{E}_{\Pr_{\mathcal{D}}} h(x) - \mathbb{E}_{\Pr_{\mathcal{D}'}}h(x)],
\end{array}
\label{eq:empirical-dw1}
\end{equation}
where $\|h\|_{\mathcal{L} \leq 1}$ represents all 1-Lipschitz functions $h : \mathcal{X} \mapsto \mathbb{R}$. Notice that, without loss of generality, we assume $\mathcal{X} = [0,1]^n$ where $n$ is the dimension of the input. Thereby, we can further assume that $h(x)\in [0,1] $ which will not change the maximum value of Eq.~\ref{eq:empirical-dw1}. Under the above assumptions, comparing with $d_{\mathcal{H}-W1}$ (Eq.~\ref{eq:h-w1}), $d_{W1}$ (Eq.~\ref{eq:empirical-dw1}) is additionally constrained by that $h$ should be a 1-Lipschitz function. Hence, $d_{W1}(\mathcal{D},\mathcal{D'}) \leq d_{\mathcal{H}-W1}(\mathcal{D},\mathcal{D'})$ as we desired.

\subsection{Proof of Theorem~\ref{thr:backdoor_general_distance}}
\label{proof:backdoor_general_distance}

\begin{proof}
Supposing $Z_{A,\mathcal{B},t} = \int_{(x,y)}P(x,y)$ where $P(x,y)$ is defined in Eq.~\ref{eq:backdoor_distribution}, $A(\cdot)$ is the trigger function, $\mathcal{B}$ is the backdoor region and $t$ is the target label, we have:

\begin{equation}
\notag
\begin{array}{r@{\quad}l}
&d_{\mathcal{H}-W1}(\mathcal{D}_P, \mathcal{D}_{A,\mathcal{B},t}) \\
=& (1-\frac{1}{Z_{A,\mathcal{B},t}}) (1-\Pr(A(\mathcal{B}))) \\
+ &\underset{h(x,y)}{\max} (\underset{(x,y) \in A(\mathcal{B})\times \mathcal{Y}}{\int} h(x,y)( \Pr_{\mathcal{D}_P}(x,y) - \Pr_{\mathcal{D}_{A,\mathcal{B},t}}(x,y))).\\
\end{array}
\end{equation}

We split $A(\mathcal{B})\times\mathcal{Y}=C_+ \cup C_-$ where $C_+=\{(x,y): \Pr_{\mathcal{D}_P}(x,y) \ge \Pr_{\mathcal{D}_{A,\mathcal{B},t}}(x,y)\}$ and $C_-=\{(x,y): \Pr_{\mathcal{D}_P}(x,y) < \Pr_{\mathcal{D}_{A,\mathcal{B},t}}(x,y)\}$.

\begin{equation}
\notag
\begin{array}{r@{\quad}l}
&d_{\mathcal{H}-W1}(\mathcal{D}_P, \mathcal{D}_{A,\mathcal{B},t}) \\
=& (1-\frac{1}{Z_{A,\mathcal{B},t}}) (1-\Pr(A(\mathcal{B}))) + \Pr(A(\mathcal{B})) - \underset{(x,y)\in C_-}{\int} \Pr_{\mathcal{D}_P}(x,y)\\
- & (\underset{(x,y)\in A(\mathcal{B})\times\mathcal{Y}}{\int} \Pr_{\mathcal{D}_{A,\mathcal{B},t}}(x,y) - \underset{(x,y)\in C_-}{\int} \Pr_{\mathcal{D}_{A,\mathcal{B},t}}(x,y)) \\
=& \underset{(x,y)\in C_-}{\int} (\Pr_{\mathcal{D}_{A,\mathcal{B},t}}(x,y) - \Pr_{\mathcal{D}_P}(x,y))\\
=& \underset{(x,y)\in A(\mathcal{B})\times\mathcal{Y}}{\int}  max(\Pr_{\mathcal{D}'_{A,\mathcal{B},t}}(x,y) - \Pr_{\mathcal{D}_P}(x,y), 0) \\ 
\end{array}
\end{equation}
Thus, we get what's desired. 

\end{proof}

\subsection{Proof of Theorem~\ref{thr:backdoor_distance_range}}
\label{proof:backdoor_distance_range}

\begin{proof}

Let's consider two Lemmas first.

\begin{lemma}~\label{lemma:sum_1}
Supposing there are two sets of $n$ numbers: $\{u_1, u_2, ..., u_n\}$ and $\{v_1, v_2, ..., v_n\}$, if $\forall {a_i} \geq 0$, $\forall {b_i} \geq 0$, $\overset{n}{\underset{i=1}{\sum}} u_i = 1$ and $\overset{n}{\underset{i=1}{\sum}} v_i = 1$, for a number $K \geq 1$, we have 
\begin{equation}
    K \geq \overset{n}{\underset{i=1}{\sum}} \max\{K u_i - v_i,0\} \geq (K-1) + \overset{n}{\underset{i=1}{\sum}} \max\{u_i - v_i,0\}.
\end{equation}
\end{lemma}
\begin{proof}
One can easily get the desired through $\max\{K u_i - v_i,0\} \geq \max\{u_i - v_i,0\} + (K-1)u_i$ and $\max\{K u_i - v_i,0\} \le K u_i$
\end{proof}

\begin{lemma}~\label{lemma:kappa_beta}
Supposing $\kappa \geq 1$, $\frac{1}{\kappa} \leq \beta \leq 1$ and $Z_{A,\mathcal{B},t} = 1 - \frac{1}{\kappa} \Pr(B) + \beta \Pr(B)$, we have $\frac{\beta}{Z_{A,\mathcal{B},t}} \geq \frac{1}{\kappa}$.
\end{lemma}

\begin{proof}

First, we have following equations.
\begin{equation}
\begin{array}{r@{\quad}l}
& \kappa \beta - Z_{A,\mathcal{B},t} \\
= & \frac{1}{\kappa}(\kappa^2 \beta - \kappa \beta \Pr(\mathcal{B}) - \kappa + \Pr(\mathcal{B})) \\
= & \frac{1}{\kappa}(\kappa - \Pr(\mathcal{B})(\kappa \beta - 1) \\
\end{array}
\end{equation}
Since $\frac{1}{\kappa} \leq \beta$, we have $\kappa \beta \geq 1$. Besides, since $\kappa \geq 1$, we have $\kappa - \Pr(\mathcal{B}) > 0$.
Putting them together, we have $\kappa \beta \geq Z_{A,\mathcal{B},t}$, or $\frac{\beta}{Z_{A,\mathcal{B},t}} \geq \frac{1}{\kappa}$ as desired.
\end{proof}

Now, we use above two Lemmas to prove Theorem~\ref{thr:backdoor_distance_range}.
First, we denote $u(x,y)$ as $\frac{\Pr_{\mathcal{D}_{A,\mathcal{B},t}}(x)}{\Pr_{\mathcal{D}_{A,\mathcal{B},t}}(A(\mathcal{B}))}\Pr_{\mathcal{D}_{A,\mathcal{B},t}}(y|x)$ and $v(x,y)$ as $\frac{\Pr(x)}{\Pr(A(\mathcal{B}))}\Pr_{\mathcal{D}_P}(y|x)$. Apperently, $u(x,y) \geq 0$ and $v(x,y) \geq 0$. Besides, $\underset{(x,y) \in A(\mathcal{B}) \times \mathcal{Y}}{\int}  u(x,y) = 1$ and $\underset{(x,y) \in A(\mathcal{B}) \times \mathcal{Y}}{\int} v(x,y) = 1$.

According to Lemma~\ref{lemma:backdoor_A_beta}, we have
\begin{equation}
\notag
\begin{array}{l@{\quad}l}
&d_{\mathcal{H}-W1}(\mathcal{D}_P, \mathcal{D}_{A,\mathcal{B},t}) \\
= &\Pr(\mathcal{B}) \underset{(x,y) \in A(\mathcal{B}) \times \mathcal{Y}}{\int} \max\{\frac{\beta}{Z_{A,\mathcal{B},t}} u(x,y)-\frac{1}{\kappa}v(x,y), 0\} \\
= & \frac{1}{\kappa} \Pr(\mathcal{B}) \underset{(x,y) \in A(\mathcal{B}) \times \mathcal{Y}}{\int} \max\{K u(x,y)-v(x,y), 0\}, \\
\end{array}
\end{equation}
\noindent where we set $K = \frac{\beta}{Z_{A,\mathcal{B},t}} / \frac{1}{\kappa}$.
According to Lemma~\ref{lemma:kappa_beta}, we have $\frac{\beta}{Z_{A,\mathcal{B},t}} \geq \frac{1}{\kappa}$ and, thus, $K \geq 1$. Applying Lemma~\ref{lemma:sum_1}, we further get
\begin{equation}
\notag
\begin{array}{l@{\quad}l}
&d_{\mathcal{H}-W1}(\mathcal{D}_P, \mathcal{D}_{A,\mathcal{B},t}) \\
\geq & \frac{1}{\kappa}\Pr(\mathcal{B}) ( (K-1)+  \underset{(x,y) \in A(\mathcal{B}) \times \mathcal{Y}}{\int} \max\{u(x,y)-v(x,y), 0\} )\\
= & \frac{1}{\kappa}\Pr(\mathcal{B}) ( (K-1)+ S).\\
\end{array}
\end{equation}
Taking $K = \frac{\beta}{Z_{A,\mathcal{B},t}} / \frac{1}{\kappa}$ into the last equation, we have 
\begin{equation}
\notag
\begin{array}{l@{\quad}l}
 (\frac{\beta}{Z_{A,\mathcal{B},t}} - \frac{1}{\kappa}(1- S)) \Pr(\mathcal{B}) \leq d_{\mathcal{H}-W1}(\mathcal{D}_P, \mathcal{D}_{A,\mathcal{B},t})
\end{array}
\end{equation}
as desired.
Similarly, we get $d_{\mathcal{H}-W1}(\mathcal{D}_P, \mathcal{D}_{A,\mathcal{B},t}) \leq \frac{\beta}{Z_{A,\mathcal{B},t}}\Pr(\mathcal{B})$. This completes this proof.

\end{proof}

\subsection{Proof of Corollary~\ref{coro:beta_effect}}
\label{proof:beta_effect}

\begin{proof}

After calculation, we get that the derivative of $\frac{\beta}{Z_{A,\mathcal{B},t}}$ w.r.t. $\beta$ is $\frac{1}{Z_{A,\mathcal{B},t}^2} (1-\frac{1}{\kappa}\Pr(\mathcal{B}))$. Since $\kappa \geq 1$, we have $\frac{1}{\kappa}\Pr(\mathcal{B}) \leq 1$. Thus, the derivative is non-negative, which indicates that $\frac{\beta}{Z_{A,\mathcal{B},t}}$ increases along with the increasing $\beta$. Considering that $\beta \in [\frac{1}{\kappa}, 1]$, we obtain that $\frac{\beta}{Z_{A,\mathcal{B},t}}$ achieves the lower-bound $\frac{1}{\kappa}$ when $\beta = \frac{1}{\kappa}$, and 
achieves the upper-bound $\frac{\kappa \Pr(\mathcal{B})}{\kappa + \kappa\Pr(\mathcal{B}) - \Pr(\mathcal{B})}$ when $\beta = 1$. Taking these results into Theorem~\ref{thr:backdoor_distance_range}, we get what's desired.

\end{proof}

\subsection{Proof of Corollary~\ref{coro:kappa_effect}}
\label{proof:kappa_effect}

\begin{proof}

Let's first consider the upper-bound of $d_{\mathcal{H}-W1}(\mathcal{D}_P, \mathcal{D}_{A,\mathcal{B},t})$, that is $\Pr(\mathcal{B}) \frac{\beta}{Z_{A,\mathcal{B}, t}}$ according to Theorem~\ref{thr:backdoor_distance_range}. After calculation, we get that the derivative of $\frac{\beta}{Z_{A,\mathcal{B}, t}}$ w.r.t. $\kappa$ is $\frac{-\beta}{Z_{A,\mathcal{B}, t}^2} < 0$. Thus, $\frac{\beta}{Z_{A,\mathcal{B}, t}}$ increases monotonously along with the decreasing of $\kappa$. Considering $\frac{1}{\kappa} \leq \beta$, we obtain that $\frac{\beta}{Z_{A,\mathcal{B}, t}}$ reaches its upper-bound $\beta$ when $\kappa=\frac{1}{\beta}$, and thus $d_{\mathcal{H}-W1}(\mathcal{D}_P, \mathcal{D}_{A,\mathcal{B},t}) \le \beta \Pr(\mathcal{B})$

Let's now consider the lower-bound of $d_{\mathcal{H}-W1}(\mathcal{D}_P, \mathcal{D}_{A,\mathcal{B},t})$, that is $\Pr(\mathcal{B}) (\frac{\beta}{Z_{A,\mathcal{B}, t}} - \frac{1}{\kappa}(1-S))$ according to Theorem~\ref{thr:backdoor_distance_range}.
However, $\frac{1}{\kappa} S$ is $o(\frac{1}{\kappa})$ when $d_{\mathcal{H}-W1}(\mathcal{D}_P, \mathcal{D}_{A,\mathcal{B},t})$ becomes close to its lower-bound, since $S \to 0$ in this case. Thus, we only need to consider the relation between $\frac{\beta}{Z_{A,\mathcal{B}, t}} - \frac{1}{\kappa}$ and $\kappa$. Particularly, we have that $\frac{\beta}{Z_{A,\mathcal{B}, t}} - \frac{1}{\kappa} = 0$ when $\kappa = \frac{1}{\beta}$.

After calculation, we get the derivative of $\frac{\beta}{Z_{A,\mathcal{B}, t}} - \frac{1}{\kappa}$ w.r.t. $\kappa$ is
\begin{equation}
\notag
\begin{array}{l@{\quad}l}
\frac{\kappa^2(1+\beta^2\Pr(\mathcal{B})^2 + \beta \Pr(\mathcal{B})) - 2\kappa (\Pr(\mathcal{B}) + 2 \beta \Pr(\mathcal{B})^2) + \Pr(\mathcal{B})^2}{Z_{A,\mathcal{B},t}^2 \kappa^2}.
\end{array}
\end{equation}
\noindent The denominator is strictly positive and numerator is a quadratic function of $\kappa$. After calculation, we get its two roots $r1$ and $r2$:
\begin{equation}
\notag
\begin{array}{l@{\quad}l}
r1 = \frac{1}{\beta} - \frac{1 - \sqrt{\beta^3 \Pr(\mathcal{B})^3}}{\beta (1+\beta^2\Pr(\mathcal{B})^2 + \beta \Pr(\mathcal{B}))}\\
r2 = \frac{1}{\beta} - \frac{1 + \sqrt{\beta^3 \Pr(\mathcal{B})^3}}{\beta (1+\beta^2\Pr(\mathcal{B})^2 + \beta \Pr(\mathcal{B}))}\\
\end{array}
\end{equation}
\noindent Apparently, $r1 < \frac{1}{\beta}$ and $r2 < \frac{1}{\beta}$. Considering the coefficient of the quadratic term is positive, we obtain that $\frac{\beta}{Z_{A,\mathcal{B}, t}} - \frac{1}{\kappa}$ increases monotonously along with the increasing $\kappa$ when $\kappa \geq \frac{1}{\beta}$. This indicates that $\frac{\beta}{Z_{A,\mathcal{B}, t}} - \frac{1}{\kappa}$ reaches its lower-bound $0$ when $\kappa = \frac{1}{\beta}$. In this case, 
$d_{\mathcal{H}-W1}(\mathcal{D}_P, \mathcal{D}_{A,\mathcal{B},t})$ reaches the minimum value $\beta S \Pr(\mathcal{B})$ as desired.


\end{proof}

\subsection{Proof of Lemma~\ref{lemma:weight_distance}}
\label{proof:weight_distance}

\begin{proof}
Supposing $f_P=c_P \circ g_P$ and $f_b=c_b \circ g_b$, we have
\begin{equation}
\notag
\begin{array}{r@{\quad}l}
\alpha^2 =& (\frac{\beta}{m L} \underset{x \in X}{\sum}   \underset{y \in \mathcal{Y}}{\sum}  max(g_P(x)_y-g_b(x)_y,0))^2 \\
\leq & (\frac{\beta}{m L } \underset{x \in X}{\sum} \underset{y \in \mathcal{Y}}{\sum}  |g_P(x)-g_b(x)|)^2 \\
\leq & (\frac{\beta}{m} \underset{x \in X}{\sum} \frac{1}{\sqrt{L}}\|g_P(x)_y-g_b(x)_y\|_2)^2 \\
\leq & \frac{\beta^2}{m L} \underset{x \in X}{\sum} \|g_P(x)-g_b(x)\|_2^2 \\
= & \frac{\beta^2}{m L} \|\phi(X)(\omega_b-\omega_P)\|_2^2 \\
\leq & \frac{\beta^2}{m L} \|\phi(X)\|_2^2 \|\omega_b-\omega_P\|_2^2.
\end{array}
\end{equation}
The inequality of arithmetic and geometric means are used to obtain the third and forth transformations. The cauchy-schwarz inequality are used to obtain the last transformation. After simple math, the lower-bound of $ \|\omega_b-\omega_P\|_2$ will be derived as what's desired.

\end{proof}

\subsection{Proof of Lemma~\ref{lemma:normal_dhw1}}
\label{proof:normal_dhw1}

\begin{proof}

Specifically, a test statistic $T^2$ is calculated as
\begin{equation}
\notag
\begin{array}{r@{\quad}l}
T^2&=\frac{n_P n_{tg}}{n_P+n_{tg}}d_M(m_P,m_b)^2 \\
&\leq \frac{n_P n_{tg}}{n_P+n_{tg}} \lambda_{max} \|m_P-m_b\|_2^2,
\end{array}
\end{equation}
where $d_M(m_P,m_b)=\sqrt{(m_P-m_b)^T\Sigma^{-1}(m_P-m_b)}$ is the Mahalanobis distance~\cite{mclachlan1999mahalanobis} and $\lambda_{max}$ is the largest eigenvalue of $\Sigma^{-1}$. 

Next, we demonstrate that $\|m_P-m_b\|_2 \leq d_{\mathcal{H}-W1}(\mathcal{N}_P,\mathcal{N}_b)$ when $\mathcal{N}_P = \mathcal{N}(m_P, \sigma)$ and $\mathcal{N}_b = \mathcal{N}(m_b, \sigma)$.
Actually, if $\|m_P-m_b\|_2 = d_{W1}(\mathcal{N}_P, \mathcal{N}_b)$, we can easily prove the inequality according to Proposition~\ref{prop:h-w1} that illustrates $d_{W1}(\mathcal{N}_P,\mathcal{N}_b) \leq d_{\mathcal{H}-W1}(\mathcal{N}_P,\mathcal{N}_b)$.
Next, we strictly prove $\|m_P-m_b\|_2 = d_{W1}(\mathcal{N}_P, \mathcal{N}_b)$.

According to the Jensen's inequality~\cite{ruel1999jensen}, $\mathbb{E}\|x-x'\|_2 \geq \|\mathbb{E}(x-x')\|_2 = \|m_P-m_b\|_2$. Thus $d_{W1}(\mathcal{N}_P, \mathcal{N}_b) \geq \|m_P-m_b\|_2$. 
Again by Jensen's inequality, $(\mathbb{E}\|x-x'\|_2)^2 \leq \mathbb{E}\|x-x'\|_2^2$. Thus, $d_{W1}(\mathcal{N}_P, \mathcal{N}_b) \leq d_{W2}(\mathcal{N}_P, \mathcal{N}_b)$ where $d_{W2}(\cdot,\cdot)$ is the Wasserstein-2 distance. As proved in paper~\cite{dowson1982frechet}, the Wasserstein-2 distance between two normal distribution can be calculated by:
\begin{equation}
\notag
\begin{array}{r@{\quad}l}
d_{W2}^2(\mathcal{N}_P, \mathcal{N}_b)
= \|m_P-m_b\|_2^2+ tr(\Sigma_P+\Sigma_b-2(\Sigma_P\Sigma_b)^{\frac{1}{2}})
\end{array}
\end{equation}
Because $\Sigma_P=\Sigma_b=\Sigma$ as we assumed, $d_{W2}^2(\mathcal{N}_P, \mathcal{N}_b) = \|m_P-m_b\|_2^2$. 
Thus, putting the above together, we get $\|m_P-m_b\|_2 \leq d_{W1}(\mathcal{N}_P, \mathcal{N}_b) \leq \|m_P-m_b\|_2$, which indicates $\|m_P-m_b\|_2 = d_{W1}(\mathcal{N}_P, \mathcal{N}_b)$ as desired.

Due to that $d_{\mathcal{H}-W1}(\mathcal{D}_P, \mathcal{D}_{A,\mathcal{B},t})$ is the maximum value among all possible separation functions, we have
\begin{equation}
\notag
\begin{array}{r@{\quad}l}
d_{\mathcal{H}-W1}(\mathcal{D}_P, \mathcal{D}_{A,\mathcal{B},t}) \geq d_{\mathcal{H}-W1}(X_P, X_b) =  d_{\mathcal{H}-W1}(\mathcal{N}_P,\mathcal{N}_b).
\end{array}
\end{equation}
Thus $\alpha \geq \|m_P-m_b\|_2$. Taking this into our original inequality, we get $T^2 \leq \frac{n_P n_{tg}}{n_P+n_{tg}} \lambda_{max}  \alpha^2$ as desired.

\end{proof}

\section{Appendix of TSA on Backdoor Unlearning}
\label{app:backdoor_unlearning}

In addition to detection, the defender may also want to remove the backdoor from an infected model, either after detecting the model or through ``blindly'' unlearning the backdoor should it indeed be present in the model.  

We classify unlearning methods for backdoor removal into two categories:  targeted unlearning (for removing detected backdoors) and ``blind'' unlearning. 

\vspace{3pt}\noindent\textbf{``Blind'' unlearning}.
Such unlearning methods can be further classified into two sub-categories: fine-tuning and robustness enhancement. The former fine-tunes a given model on benign inputs, through which Catastrophic Forgetting (CF) would be induced so an infected model's capability to recognize the trigger may be forgotten. To study the relationship between CF and the backdoor similarity, we identify the lower bound of \textit{task drift} based on Lemma~\ref{lemma:weight_distance}:
\begin{equation}
\label{eq:task_drift_lower_bound}
\begin{array}{r@{\quad}l}
\|\delta^{\mathcal{T}_P \to \mathcal{T}_{A,\mathcal{B},t}}(X)\|_2 \geq \alpha \frac{\sqrt{mL}}{\beta}
\end{array}
\end{equation}

Eq~\ref{eq:task_drift_lower_bound} shows that small \textit{task drift}, the measurement of CF, requires small backdoor distance (depicted by $\alpha$) between the primary task $\mathcal{T}_P$ and the backdoor task $\mathcal{T}_{A,\mathcal{B},t}$, implying that ``blind" unlearning through fine-tuning becomes less effective (i.e., the backdoor may not be completely forgotten) when the backdoor distance is small.


The robustness enhancement methods aim to enhance the robustness radius of a backdoor model $f_b$ within which the model prediction remains the same. Specifically, the robustness radius $\underset{-}{\bigtriangleup}(X,s)$ for the source label $s$ on a set of benign inputs $X$ could be formulated as
\begin{equation}
\notag
\begin{array}{r@{\quad}l}
\underset{-}{\bigtriangleup}(X,s) \overset{def}{=} \underset{x \in X_{f_b(x)=s}}{\min} \{\bigtriangleup(x) : \underset{f(x+\delta) \neq s}{\inf} \|\delta\| \}.
\end{array}
\end{equation}
We denote $R(X,s)$ as the set of inputs $x'$ within the robustness radius $\underset{-}{\bigtriangleup}(X,s)$, 
\begin{equation}
\notag
\begin{array}{r@{\quad}l}
R(X,s) = \{x': \underset{x \in X_{f_b(x)=s}}{\inf} \|x'-x\| < \underset{-}{\bigtriangleup}(X,s)\}.
\end{array}
\end{equation}
Clearly, when $\underset{-}{\bigtriangleup}(X,s)$ increases, $R(X,s)$ becomes larger. However, increasing $\underset{-}{\bigtriangleup}(X,s)$ is less effective for removing backdoors with small backdoor distances $d_{\mathcal{H}-W1}(\mathcal{D}_P, \mathcal{D}_{A,\mathcal{B},t})$ for the following reasons.
1) When $R(X,s) \cap A(\mathcal{B}) = \emptyset$, apparently, the predicted labels of trigger-carrying inputs 
do not change. 
2) When $R(X,s) \cap A(\mathcal{B}) \neq \emptyset$ and $A(\mathcal{B})\setminus R(X,s) \neq \emptyset$, the small $d_{\mathcal{H}-W1}(\mathcal{D}_P, \mathcal{D}_{A,\mathcal{B},t})$ will lead to a large $A(\mathcal{B})\setminus R(X,s)$, i.e., the more $x \in A(\mathcal{B})$ close to the decision boundary, the more $x \in A(\mathcal{B})$ outside $R(X,s)$, indicating that the backdoor remains largely un-removed. This is because, during robustness enhancement, $f_b$ is learned to push $x \in X$ away from the boundary as much as possible, which is considered over-fitting by the neural network.
3) When $A(\mathcal{B})\setminus R(X,s) = \emptyset$, $R(X,s)$ covers many inputs within the robust radius whose true label is not $s$, i.e., $f^*(x') \neq s$, and thus the robustness enhancement will result in a false prediction on these inputs, which is not desired. This is due to the irregular classification boundary of $f_b$ that makes the precise removal of the backdoor impossible without knowing the trigger function $A$. Besides, increasing $\underset{-}{\bigtriangleup}(X,s)$ will decrease $\underset{-}{\bigtriangleup}(X,t)$ for $t \neq s$, which eventually results in a model making the false prediction on the inputs with the true label of $t$.

\vspace{3pt}\noindent\textbf{Targeted unlearning}.
The targeted unlearning methods are guided by the triggers reconstructed by the backdoor detection methods. As we demonstrated in Section~\ref{sec:detection}, backdoor detection methods themselves become hard when the backdoor distance is small. Therefore, the targeted unlearning methods also become less effective for the backdoors with smaller backdoor distances.

\ignore{
\subsection{Experimental Analysis}

To study how hard to unlearn the backdoor with various $\alpha$, we performed fine-tuning on those backdoored model. Specifically, we obtained 50 backdoored models with 10 different $\alpha$ on CIFAR-10 and fine tuned these backdoored models on benign inputs until the ASR of these backdoored model be lower than $10\%$. Table~\ref{tb:fine-tune} illustrates the number of epochs needed to reach such low ASR for backdoored models with different $\alpha$.

\vspace{-5pt}
\begin{table}[htb]
\centering
\footnotesize
\caption{Epochs needed by fine-tuning to forget backdoors with different $\alpha$ (the results are averaged among 5 models).}
\vspace{-10pt}
\begin{adjustbox}{width=0.48\textwidth}
\begin{tabular}{lllllllllll}
\hline
                & \textbf{0.1}      & \textbf{0.2}      & \textbf{0.3}      & \textbf{0.4}      & \textbf{0.5} & \textbf{0.6} & \textbf{0.7} & \textbf{0.8} & \textbf{0.9} & \textbf{1.0} \\ \hline
\textbf{Epochs} & \textgreater{}300 & \textgreater{}300 & \textgreater{}300 & \textgreater{}300 & 260          & 192          & 153          & 93.4         & 40.4         & 11.2     \\ \hline   
\end{tabular}
\end{adjustbox}
\label{tb:fine-tune}
\end{table}
\vspace{-5pt}
}

\section{Appendix of Backdoor Disabling}
\label{app:backdoor_disabling}



Even though the backdoor with small backdoor distance is hard to be detected and unlearned from the target model, the defender could suppress the backdoor behaviour through backdoor disabling methods. 
Backdoor disabling aims to remove the backdoor behaviour of infected model without affecting model predictions on benign inputs. There are mainly two kinds of methods: knowledge distillation and inputs preprocessing.

\vspace{3pt}\noindent\textbf{Knowledge distillation}.
In knowledge distillation, usually, there is a teacher model and a student model. The knowledge distillation defense uses the knowledge distillation process to suppress the student model from learning the backdoor behaviour from the teacher model through temperature controlling. Specifically, following the notion used in paper~\cite{distillation_defense}, for the temperature $T=1$, we have $f_b(x)_j = \textit{softmax}_{T=1}(\mu_j)$ where $\textit{softmax}_T(\mu_j) = \frac{\exp(\mu_j)/T}{\sum_{j' \in \mathcal{Y}} \exp(\mu_{j'}/T)}$. 
The bigger is $T$, the softer is the prediction result $f_b(x)$. The high temperature (e.g., $T=20$ as used by ~\cite{distillation_defense}) could prevent the student model from learning typical backdoors that drives the model to generate highly confident predictions (of the target label) on trigger-carrying inputs. However, backdoors with low backdoor distance drive the model to generate only moderate predictions for trigger-carrying inputs that are close to the classification boundary, which may still be learned by the student model through the high temperature knowledge distillation. Besides, the higher is the temperature, the smaller amount of knowledge could be learned by the student model, which result in the relatively low accuracy of the student model. 
On the other hand, the low temperature will sharpen the classification boundary, which allows the student model to learn confident predictions from teacher model. However, in this case, the predictions of the teacher model for trigger-carrying inputs become confident, i.e., $f_b(A(x))_t$ is high. As a result, the student model may easily learn the backdoor, in a similar manner as learning it from a contaminated training dataset in a traditional backdoor attack when the backdoor has a large backdoor distance. Therefore, the choice of the temperature reflects a trade-off between the performance (i.e., the effectiveness of knowledge distillation) and security (i.e., the effectiveness of backdoor disabling) of the student model. 
Finally, the knowledge distillation could be viewed as a continual learning process from the backdoor task $\mathcal{T}_{A,\mathcal{B},t}$ to the primary task $\mathcal{T}_P$. As shown in Eq.~\ref{eq:task_drift_lower_bound}, it is expected that the student model will give similar predictions as the teacher model in predictions of either clean or trigger-carrying inputs, when the backdoor has a small backdoor distance. 

\vspace{3pt}\noindent\textbf{Input preprocessing}.
This kind of defenses introduce a preprocessing module before feeding the inputs into the target model that removes the trigger contained in inputs~\cite{li2022backdoor}. Accordingly, the modified triggers no longer match the hidden backdoor and therefore preventing the activation of the backdoor. Without knowing the details of the trigger, these methods perform preprocessing on both benign inputs and trigger-carrying inputs. Thus, actually, these methods disable backdoor based on a fundamental assumption that the trigger is sensitive to noise and the robustness of benign model $f_P$ and backdoor model $f_b$ for trigger-carrying inputs differ significantly.
To study the robustness of backdoors with small backdoor distance,  we investigate the difference between the predictions for the trigger-carrying inputs and the trigger-carrying inputs with small added noise $\delta$, i.e., $|f_b(A(x)+\delta)_t-f_b(A(x))_t|$. When the backdoor distance is small, not only $A(x)$ is close to the classification boundary of the backdoor model $f_b$ but also close to the classification boundary of benign model $f_P$. Intuitively, $A(x)+\delta$ is close to the classification boundary of both $f_b$ and $f_P$ when $\|\delta\|$ is small. Thus, $|f_b(A(x)+\delta)_t-f_b(A(x))_t|$ should be small. One may argue that some $\delta$ would make $\underset{j}{\argmax} f_b(A(x)+\delta)_j \neq t = \underset{j}{\argmax}f_b(A(x))_j$ even if $|f_b(A(x)+\delta)_t-f_b(A(x))_t|$ is small, as small $\delta$ could flip the predicted label for $A(x)$ that is close to the classification boundary in $f_b$. However, for the same $\delta$, the benign model would also flip the predicted label for $A(x)+\delta$, i.e., $\underset{j}{\argmax} f_P(A(x)+\delta)_j \neq \underset{j}{\argmax}f_P(A(x))_j$, because $A(x)$ is also close to the classification boundary in $f_P$. There is no reason to block the input which is predicted by the backdoor model $f_b$ with the same label as the one predicted by a benign model $f_P$.
If a large noise $\delta$ was added to the inputs, the performance of the deep neural networks (both $f_P$ and $f_b$) will decrease. Consequently, even though the backdoor is suppressed, the benign model for the primary task will become worse, indicating a trade-off between utility and security, which will be discussed in the next section. 
In general, there is no significant difference between the robustness of the benign model $f_P$ and the backdoor model $f_b$ for trigger-carrying inputs, when backdoor distance is small. Thus, the input preprocessing methods may only moderately suppress the backdoors with small backdoor distances.

\section{Appendix of IMC Experiments}
\label{app:IMC}

We exploited IMC to generate 200 backdoored models on CIFAR10 with its official code that has been integrated into the TrojnZoo framework. In this experiment, those backdoors carried on those backdoored models are source-specific with the source class is 1 and the target class is 0. We set $\beta=0.1$, i.e., injecting 500 trigger-carrying inputs into the training set, and kept other parameters as the default values, i.e., the trigger size is 3x3 and the transparency of trigger is 0 (meaning the trigger is clear and has not been blurred). 
Table~\ref{tb:IMC} illustrates the accuracy of 6 detections in distinguishing 200 IMC backdoored models from 200 benign models, comparing with what has been obtained by TSA attacks.

\begin{table}[tbh]
\vspace{-5pt}
  \centering
  \caption{The backdoor detection accuracies (\%) for IMC and TSA obtained by six detection methods.}
  \begin{adjustbox}{width=0.40\textwidth}
\begin{tabular}{|c|c|c|c|c|c|c|}
\hline
    & K-ARM & MNTD  & ABS   & TND   & SCAn  & AC    \\ \hline
IMC & 89.50 & 98.75 & 74.25 & 80.25 & 91.00 & 73.50 \\ \hline
TSA & 59.25 & 51.25 & 51.00 & 48.75 & 63.25 & 55.25 \\ \hline
\end{tabular}

\end{adjustbox}
	\label{tb:IMC}
\vspace{-5pt}
\end{table}

\end{document}